\documentclass[11pt]{amsart}
\usepackage{times}
\usepackage{amsfonts,amssymb,amsmath,amsgen,amsthm,faktor}
\usepackage{mathtools}
\usepackage{hyperref}
\usepackage[T1]{fontenc}
\usepackage[utf8]{inputenc}
\usepackage[american]{babel}
\usepackage{cancel}
\usepackage[a4paper]{geometry}
\usepackage{multicol}
\usepackage[pdftex]{color,graphicx}
\usepackage{bbm}
\usepackage{enumerate}
\usepackage[normalem]{ulem}
\usepackage{comment}

\theoremstyle{definition}
\newtheorem{definition}{Definition} 

\theoremstyle{definition}
 
\newtheorem{theorem}[definition]{Theorem}

\theoremstyle{remark}
\newtheorem{remark}{Remark}[section]

\newcommand{\A}{\mathfrak{A}}
\newcommand{\nalgebra}{\mathfrak{M}}
\newcommand{\M}{\mathcal{M}}
\newcommand{\eps}{\varepsilon}
\newcommand{\R}{\mathbb{R}}
\newcommand{\N}{\mathbb{N}}
\newcommand{\op}{{\rm op}}
\newcommand{\Z}{\mathbb{Z}}

\newcommand{\<}{\left\langle}
\newcommand{\Tr}{\text{Tr}}
\newcommand{\C}{\mathbb{C}}
\newcommand{\K}{\mathcal{K}}
\newcommand{\UM}{\mathbbm{1}}

\renewcommand{\L}{\mathcal{L}}
\renewcommand{\H}{\mathcal{H}}
\renewcommand{\>}{\right\rangle}
\newcommand{\ip}[2]{\left\langle #1 , #2 \right \rangle}
\newcommand{\co}[1]{co\left(#1\right)}
\newcommand{\cco}[1]{\overline{co}\left(#1\right)}

\newcommand{\calF}{\mathcal{F}}
\newcommand{\bbE}{\mathbb{E}}
\newcommand{\Var}{\mathrm{Var}}
\newcommand{\biglan}{\big\langle}
\newcommand{\bigran}{\big\rangle}
\newcommand{\Biglan}{\Big\langle}
\newcommand{\Bigran}{\Big\rangle}
\newcommand{\PTr}{\mathrm{Ptr}}
\newcommand{\Cstar}{$\mathrm{C}^\ast$}

\begin{document}
\title{Pure and Mixed States}
\author{J. C. A. Barata}
\address[J.\ C.\ A.\ Barata]{Instituto de Física, Universidade de São Paulo,  CP 66.318, 05314-970, São Paulo, SP, Brazil.}
\email{jbarata@if.usp.br}
\author{M.\ Brum}
\thanks{M.\ Brum was supported by grant 2015/02975-4, São Paulo Research Foundation (FAPESP)}
\address[M. Brum]{Departamento de Matemática, Universidade Federal do Rio de Janeiro -- Campus Duque de Caxias. Duque de Caxias, RJ, 25265-970, Brazil.}
\email{marcos.brum@xerem.ufrj.br}
\author{V.\ Chabu}
\thanks{V.\ Chabu was supported by grant 2017/13865-0, São Paulo Research Foundation (FAPESP)}
\address[V.\ Chabu]{Instituto de Física, Universidade de São Paulo, CP 66.318, 05314-970, São Paulo, SP, Brazil.}
\email{vbchabu@if.usp.br}
\author{R.\ Correa da Silva}
\address[R.\ Correa da Silva]{Instituto de Física, Universidade de São Paulo,  CP 66.318, 05314-970, São Paulo, SP, Brazil.}
\email{ricardo.correa.silva@usp.br}

\begin{abstract}
  We present a review on the notion of pure states and mixtures as
  mathematical concepts that apply for both classical and quantum
  physical theories, as well as for any other theory depending on
  statistical description. Here, states will be presented as
  expectation values on suitable algebras of observables, in a manner
  intended for the non-specialist reader; accordingly, basic
  literature on the subject will be provided. Examples will be exposed
  together with a discussion on their meanings and implications. An
  example will be shown where a pure quantum state converges to a
  classical mixture of particles as Planck's constant tends to zero.
\end{abstract}

\maketitle


\section{Introduction}\label{sec:introduction}

In many textbooks on quantum physics there is some obscurity
surrounding the notion of pure state, often giving rise to some
misconceptions. In quantum mechanics, for instance, pure states are
frequently associated to normalised vectors (or, more precisely, to
rays) in adequate Hilbert spaces. This is neither an adequate
definition nor is, strictly speaking, a correct one since, according
to a mathematical construction known as GNS representation, any state
(including mixed ones) in an adequately defined algebra of observables
can be represented by a vector state in some Hilbert space. Moreover,
this pseudo-definition fails to capture the statistical quality of the
notions of pure and of mixed states, which is actually quite simple
and illuminating.

Our objective in this review is to present both the intuitive meaning
of the concept of pure and of mixed states as well as to develop the
mathematical (algebraic) formalism around these notions in order to
clarify some of these issues.

Besides, we will explore this formalism in order to enhance our
understanding about the physical nature hidden beneath these algebraic
and statistical notions. One of the most notable results in this
direction is the possibility of unifying the treatment of quantum and
classical observables as elements in the same kind of algebra, known
as \Cstar-algebra, the only difference being that in the classical
case these elements commute with respect to the algebra's operation
(multiplication or composition), whereas in the quantum case they do
not necessarily commute.

We also present in Section \ref{sec:quantum_classical} a systematic
procedure (known as Weyl quantisation) for transforming classical
observables, namely, functions on a phase space, into quantum ones,
\emph{i.e.}, self-adjoint operators acting on a Hilbert space. Then,
by means of tools provided by Semi-Classical Analysis we will
introduce, we will exhibit an example of a family of \underline{pure}
quantum states that degenerates into classical \underline{mixtures} in
the limit when Planck's constant is taken small.

The physically important relation between purity of states and
irreducibility of certain representations of the algebra of
observables is discussed in Section \ref{sec:irreducibility}.  The
notion of purification of states is presented in Section
\ref{sec:Purification-of-States}.  Sections \ref{sec:normal_states},
\ref{sec:moreissuesaboutpurity} and
\ref{sec:Krein-MilmanandChoquettheorems} are more technical,
mathematically. In Section \ref{sec:normal_states} we introduce the
important notion of normal state and present its relation to the
so-called density matrix, a very important notion in Quantum Mechanics
relevant in discussions concerning certain aspects of Quantum
Information Theory, as the notion of von Neumann entropy and
entanglement entropy. In Section \ref{sec:moreissuesaboutpurity} we
present some interesting remarks about the notion of purity and in
Section \ref{sec:Krein-MilmanandChoquettheorems} we present results,
through Krein-Milman and Choquet's Theorems, on the existence of pure
states and on the decomposition of general states into pure ones.


\subsection{Pure and mixed probability
  distributions}\label{sec:ProbabilityDistributions}

We start our presentation considering the simpler and perhaps more
familiar context of probability distributions, where the notions of
purity and of mixture can be discussed in a quite elementary way.

Let us consider a probability space, which consists of a set $\Omega$,
called {\em event space}, and a family $\calF$ of subsets of $\Omega$,
called {\em events}. For technical reasons $\calF$, has to be a
$\sigma$-algebra of sets, but this point will not be relevant on what
follows, except to point out that $\Omega$ itself and the empty set
$\emptyset$ are possible events, that means, are elements of $\calF$.

A probability measure $\mu$ in $(\Omega,\; \calF)$ is an assignment of
each event $A\in\calF$ to a real number in $[0, \; 1]$ such that the
following conditions are fulfilled: $\mu(\emptyset)=0$, $\mu(\Omega)=1$
and $\mu\left(\bigcup_{n\in\N}A_n\right)=\sum_{n=1}^\infty \mu(A_n)$
for any collection $\{A_n\in\calF, \; n\in\N\}$ of disjoint subsets.

These postulates, widely known in the literature as {\em Kolmogorov
  axioms}, capture the essential ingredients of the intuitive notion
of probability and many basic properties of probability measures can
be directly derived from them. For instance, one of the easy
consequences of the above postulates is that $\mu(A)\leq \mu(B)$ for
any $A$ and $B$ in $\calF$ such that $A\subset B$.

As a simple example, let $\Omega = \mathbb{R}$ and let $\mu$ assign,
to any measurable subset $A \subset \mathbb{R}$ (for instance, an
open interval), the number
\[
\mu(A) = \frac{1}{\sqrt{2\pi}}\int_A \mathrm{e}^{-x^2/2} \mathrm{d}x \; .
\]
$\mu(A)$ is the probability of occurrence of event $A$ for the
particular Gaussian distribution considered in the integral.

A probability measure $\mu$ is said to be a {\em mixture} if there are
two other distinct probability measures $\mu_1$ and $\mu_2$, on the
same probability space, and numbers $\lambda_1, \; \lambda_2\in (0, \; 1)$
with $\lambda_1+\lambda_2=1$ such that
\begin{equation}\label{eq:combinacaolinearconvexa}
  \mu(A) \; = \; \lambda_1 \mu_1(A) + \lambda_2\mu_2(A)
\end{equation}
holds for all events $A\in\calF$. A probability measure is said to be
{\em pure}, or extremal, if it is not a mixture.
In the Bayesian parlance, the probabilities $\mu_1$ and $\mu_2$ are
priors of $\mu$ and $\lambda_1$ and $\lambda_2$ are their respective
likelihoods. 

An expression like \eqref{eq:combinacaolinearconvexa} is called a
convex linear combination of $\mu_1$ and $\mu_2$.  Notice that $\mu_1$
or $\mu_2$ may be mixtures themselves and, hence, we can say that a
probability measure is a mixture if it can be written as a finite (or
even infinite) convex sum of distinct probability measures:
$\mu(A)=\sum_{k=1}^n\lambda_k\mu_k(A)$, for some $n\in\N$, with
$\sum_{k=1}^n\lambda_k=1$ and $\lambda_k\in(0, \; 1)$ for all
$k$.

In order to explain the intuitive nature of a mixed probability
distribution, let us consider a very simple situation where mixture
occurs. Suppose we order a large amount of balls from two different
factories, each factory having its own standard fabrication processes.
The balls produced in each factory are not perfectly the same and will
differ randomly from each other. If we consider one specific
parameter for characterising the balls, say, their diameter, we can
associate to each factory a probability distribution associated to the
diameter: for $0<d_1<d_2$, the quantity $\mu_k\big( (d_1, \; d_2) \big)$
measures the probability for a ball produced in factory $k=1, \; 2$,
to have a diameter in the interval $(d_1, \; d_2)$.

Now, consider that we mix the balls produced in both factories, so
that a fraction $\lambda_k \in (0, \; 1)$ comes from the production of
factory $k=1, \; 2$. Naturally $\lambda_1+\lambda_2=1$. If we
measure the diameters of the balls in this mixed ensemble, it is
intuitively clear that measurements of the diameters of the balls will
be described by a probability measure $\mu$ given by
$\mu\big( (d_1, \; d_2) \big) = \lambda_1 \mu_1\big( (d_1, \; d_2)
\big)+ \lambda_2 \mu_2\big( (d_1, \; d_2) \big)$, again with $0<d_1<d_2$.

The probability $\mu$ is therefore a mixture of the probabilities
$\mu_1$ and $\mu_2$ with fractions $\lambda_1$ and $\lambda_2$,
respectively, since the ensemble considered is a mixture (in the
common sense of the word) of two ensembles described by the two
probabilities $\mu_1$ and $\mu_2$. 

Notice that the probabilities $\mu_1$ and $\mu_2$ can be themselves
mixtures, as it can happen if, for instance, the balls are produced
by different machines in each of the factories.

This example illustrates the intuitive idea behind the notion of a
mixed probability distribution: it describes samples composed by
objects of different origins which are placed together.  In contrast,
pure probability distributions describe systems that, in a sense, are
not decomposable in simpler ones. As we will see, in the case of
quantum systems these notions are neatly reproduced in the algebraic
formalism.


\subsection{Mean values and variances}\label{sec:MeanValues}

Given a probability distribution on a probability space, there is a
series of statistical quantities that provide information on the
distribution. They can also provide some insight on the nature of pure
and mixed probability distributions.

Let $f: \Omega\to\R$ be a real function defined on the event space
representing some observable quantity.  (Technically, $f$ has to be a
measurable function, but we will not stress such mathematical points
by now). We define its {\em expectation}, {\em average} or {\em mean
  value}, according to the probability measure $\mu$ by
$$
\bbE_\mu(f) \; \equiv \; \langle f\rangle_\mu \; := \;
\int_\Omega f\, d\mu
\;.
$$
The {\em variance} of $f$ on $\mu$ is defined by
$$
\Var_\mu (f)
\; := \;
\Big\langle \big( f - \langle f\rangle_\mu \big)^2 \Big\rangle_\mu
\; = \;
\langle f^2\rangle_\mu - \langle f\rangle_\mu^2 
\;.
$$
As we see from the definition, $\Var_\mu (f)$ measures how much $f$
typically deviates from its mean value $\langle f\rangle_\mu$.  It is
clear from the definition that $\mathrm{Var}_\mu (f)\geq 0$.  The
quantity $\sigma_\mu(f):=\sqrt{\mathrm{Var}_\mu (f)}$ is called the {\em
  standard deviation} of $f$ on $\mu$.

Consider a mixed probability measure
$\mu=\lambda_1\mu_1 + \lambda_2\mu_2$, with $\mu_1$ and $\mu_2$ being
two distinct probability measures in some probability space and
$\lambda_1, \; \lambda_2\in (0, \; 1)$ with
$\lambda_1+\lambda_2=1$. Then one can easily see that
$$
\langle f\rangle_\mu \; = \; \lambda_1 \langle f\rangle_{\mu_1} + \lambda_2 \langle f\rangle_{\mu_2}
\;.
$$
Moreover, one can also easily verify that
$$
\Var_\mu (f) \; = \;
\lambda_1 \Var_{\mu_1} (f) 
+ \lambda_2\Var_{\mu_2} (f)
+ \lambda_1\lambda_2
\Big[\langle f \rangle_{\mu_1} - \langle f \rangle_{\mu_2}\Big]^2
\; .
$$
From this we conclude that
$$
\Var_\mu (f)
\; \geq \;
\lambda_1 \Var_{\mu_1} (f) 
+ \lambda_2\Var_{\mu_2} (f) 
\; \geq \;
\min\big\{\Var_{\mu_1} (f)\;, \Var_{\mu_2} (f) \big\}
  \; .
$$
Hence, for the mixed probability measure $\mu$ the variance
$\Var_\mu (f)$ is always larger than or equal to the smallest of the
numbers $\Var_{\mu_1} (f)$ or $\Var_{\mu_2} (f)$. Therefore, for a
fixed function $f$, the smallest values of $\Var_\mu (f)$ will be
obtained for pure measures on this probability space. In this sense,
pure probability measures are those for which the deviation of $f$
from its mean value is smallest.


\section{The Notion of State}\label{sec:the_notion}

In Physics, the word ``state'' is often used in a somewhat informal
sense as a set $S$ of intrinsic characteristics of a system maximally
specifying the possible outcomes of measurements of observable
quantities.  A given physical theory specifies which quantities are
observable ({\em i.e.}, measurable through experiments) and a state
can be defined, with a little more precision, as a rule associating
each observable $A$ and each set $S$ of a system's physical
characteristics to a probability measure $\mu_{S, \, A}$ describing
the statistical distribution of repeated measurement of $A$ on an
ideally infinite ensemble of physical systems with the same set of
characteristics $S$. Although this definition is still vague, it is
the base for the precise definition of state that we will present
below, which is algebraic in its nature.

In Classical Mechanics, for instance, observables are (measurable)
functions $A(q, \; p)$ defined in phase space and the state of a
system is specified by a probability distribution $\rho(q, \; p)$
defined in phase space so that the mean values of repeated
measurements of $A$ in the state $\rho$ are given by
$\langle A \rangle_\rho= \int A(q, \; p)\rho(q, \; p)dqdp$.

A relevant case consists of states given by the probability
distribution $\rho_0(q, \; p)=\delta(q-q_0)\delta(p-p_0)$, where
$\delta$ represents the Dirac measure and where $(q_0, \; p_0)$ is a
given point in phase space. In this case we have
$\langle A \rangle_{\rho_0} = A(q_0, \; p_0)$. Moreover, as one easily
checks, $\Var_{\rho_0}(A)=0$, leading to the interpretation that all
individual measurements of $A$ in the state $\rho_0$ will result in
the same value $ A(q_0, \; p_0)$. Hence, the state $\rho_0$ represents a
deterministic state, fully characterised by $q_0$ and $p_0$, where
measurements of observable quantities always lead to the same result.

On the other hand, in Quantum Mechanics it is commonly thought that all the possible
pure ``states'' of a physical system
are described by normalised vectors of a Hilbert space, and the
possible measurable observables by self-adjoint operators acting on
them. Vectors in a Hilbert space $\H$ may indeed represent pure
states, (either in the intuitive notion explained above or in the
formal one to be presented below). However, as we shall discuss, not
every state can be represented as a vector in $\H$ and vector
states are nor necessarily pure.

For instance, let us consider the Hilbert space representing a
two-level system (a {\em qubit}) described in the Hilbert space
$\H = \C^2$. If we have several copies of this system in the same
state $\psi\in \H$ (with
$\|\psi\|^2=\<\psi,\,\psi\>=1$)\footnote{Scalar products (or inner
  products) in Hilbert spaces will be always denoted here by
  $\langle \phi, \; \psi\rangle$ rather than by
  $\langle \phi\,|\, \psi\rangle$. We follow the physicists'
  convention: they are antilinear in the first argument and linear in
  the second.} and measure each one of them for an observable $A$, the
mean value $\< A\>_\psi$ of the measured results is given by the inner
product $\< A \>_\psi = \<\psi,\, A\psi\>$. However, if the copies are
composed by a few systems in a state $\phi_1\in \H$ and a few other in
a different state $\phi_2\in \H$ (let us suppose a fraction $p_1$ of
the total number of particles in $\phi_1$, and $p_2$ in $\phi_2$, so
$p_1 + p_2 = 1$), then the mean value of the several measurements is
expected to be
$\< A \> = p_1 \<\phi_1, \, A\phi_1\> + p_2 \<\phi_2, \,
A\phi_2\>$. Calculating the average of the measures taken for
different copies of a system is precisely what is meant by
\emph{average value} of an observable for a system in a defined state,
so there \emph{must} be a state in which the system can be that
corresponds to this mean value
$p_1 \<\phi_1,\, A\phi_1\> + p_2 \<\phi_2,\, A\phi_2\>$.

Is there any vector state that could represent such a mixture of
states, \emph{i.e.}, a vector $\psi \in \C^2$ such that
$\<\psi,\, A\psi\> = p_1 \<\phi_1,\, A\phi_1\> + p_2 \<\phi_2,\, A\phi_2\>$?
The answer is \emph{no}, unless it is a trivial mixture where either
$p_1$ or $p_2$ is $0$, or $\phi_1 = \phi_2$ (this will follow from
Theorem \ref{th:vector_is_pure} in Section
\ref{sec:pure_and_vector}). An example of such impossibility is shown
in Section \ref{sec:example_mixture}, leading us to the conclusion
that a more comprehensive way of representing states is needed in
order to fully describe a quantum system.

Mixtures as those commented above are usually introduced in a quantum
theory based on (separable) Hilbert spaces $\H$ as \emph{density
  operators} $\rho$, which are positive trace class operators
normalised so as $\Tr(\rho) = 1$.
For a finite (possibly infinite)
mixture of states $\phi_1,\, \ldots,\, \phi_n$ with weights $p_1,\, \ldots ,\, p_n$ (and
$\sum_{k = 1}^n p_k = 1$), it is constructed as
\begin{equation}\label{eq:density_operator}
\rho = \sum_{k = 1}^n p_k \left| \phi_k \>\< \phi_k \right|
\end{equation}  
(in the infinite case, the sum's convergence is uniform) and it is
easy to see that the average value may be calculated by means of the
formula $\<A\>_\rho = \Tr(\rho A)$, since
\begin{equation*}
\Tr(\rho A) = \sum_{k = 1}^n p_k \<\phi_k,\; A\phi_k\>.
\end{equation*}

We shall denote by $\L(\H)$ the set of bounded (\emph{i.e.},
continuous) linear operators acting on a Hilbert space $\H$.

We remark that any operator such as that in equation
\eqref{eq:density_operator} has the properties listed above for
density operators. Conversely, if some $\rho \in \L(\H)$ is positive
and $\Tr(\rho)<\infty$, then it is a compact operator (see {\em e.g.}
\cite{reed_simon}) and, therefore, possesses discrete and finitely
degenerate spectrum and a spectral decomposition
$\rho = \sum_{k = 1}^\infty \varrho_k \left| \phi_k \>\< \phi_k
\right|$, where the $\phi_k$ are normalised and mutually orthogonal,
and $\varrho_k > 0$ with $\sum_{k = 1}^\infty \varrho_k = 1$, allowing
us to interpret $\rho$ as the density operator of an infinite mixture of
states $\phi_k$ with statistical weights $\varrho_k$.

\subsection{The physically motivated topology on states}\label{sec:physical_topology}

Any deeper analysis of Quantum Physics requires the introduction of
topologies in the set of observables or in the set of states, that
means, the introduction of the notion of {\em closeness} between
different observables or between different states. This is
particularly relevant if we intend to use the notion of convergence of
observables or of states. The difficulty is that the space of
observables (typically a \Cstar-algebra) and the associated set of
states are usually infinite dimensional and, therefore, are usually
equipped with many non-equivalent topologies, {\em i.e.}, with many
non-equivalent notions of closeness between their elements and,
therefore, equipped with many non-equivalent notions of convergence.

For the space of states it is possible to point to a physically
motivated topology, that we now describe. Other useful topologies will
be mentioned later (as the norm topology in the set of states in
\Cstar-algebras. See Definition \ref{def:state}).

Let us take an observable $A$; for the sake of simplicity, imagine
that it possesses a finite set of possible outcome values, and let
$\lambda$ be one of them. Then, its relative frequency
$\nu_A(\lambda,n)$ when $A$ is measured on $n$ copies of a system in
the state $\omega$ should converge to a probability $p_A(\lambda)$ in
the limit $n \rightarrow \infty$, and so must converge the
measurements' mean value $\sum_k \lambda_k \nu(\lambda_k,n)$ to the
expectation value $\omega(A)$. As we saw in the beginning of this
section, saying that a system is in the state $\omega$ means precisely
that whenever we take measurements on its copies for an observable
$A$, we will obtain an outcome according to the probability
distribution $\lambda \longmapsto p_A(\lambda)=\omega(E^A_\lambda)$,
where $E^A_\lambda$ is the spectral projection of $A$ on the closed
subspace of eigenvectors with eigenvalue $\lambda$. This is, in
theory, how to connect the positive linear functional $\omega$ to the
intrinsic property of a physical system that we called \emph{state};
in practice, this connection is trickier.

It happens that it is not possible to take an infinite number of
measurements resulting in a value $\lambda$, nor can we obtain all
possible outcomes for an observable in the case where it has an
infinite set of possible results, so we never know a system's exact
probability distributions $\lambda \longmapsto p_A(\lambda)$. Worse,
we may not measure a system for every possible observable $A$. As a
consequence, we are bent not to know a system's state exactly, for the
most we may hope to acquire from data are some approximative mean
values for small sets of finite observables.

In order to treat this problem, consider first a particular bounded
observable $A$. The first and the second issues are set by taking a
number $n$ of measurements sufficiently large so as to have:
$$\left| v(A) - \omega(A) \right| < \eps,$$
where $v(A) = \sum_{k}\lambda_k \nu_A(\lambda_k,n)$ is the measured
mean value of $A$, and $\eps > 0$ is an error that we may take
arbitrarily small, by hypothesis. Unfortunately, it is possible that a
sufficiently large $n$ that works for all observables does not
exist. Yet, we may have arbitrarily precise information on $\omega$
for finite sets of observables: taking observables
$A_1,\,\ldots ,\, A_k$ and an error $\eps$, we define a
\emph{neighbourhood} of $\omega$ with radius $\eps$ in the set of all
states $\Sigma$ as
$$
\Omega_\eps(A_1,\, \ldots ,\, A_k;\omega)
\;=\; \Big\{
\alpha \in \Sigma \; : \; \left| \alpha(A_j) - \omega(A_j) \right| < \eps,
\; \forall j \in \{1,\, \ldots \, ,k\}
\Big\}.
$$
Surely, $\omega \in \Omega_\eps(A_1,\,\ldots,\, A_k; \omega)$. Also, this
neighbourhood contains any state $\tilde{\omega}$ satisfying
$\tilde{\omega}(A) = v(A)$ at least for $A = A_1,\, \ldots,\, A_k$, so under
any measurement of observables $A_1,\, \ldots ,\, A_k$, $\tilde{\omega}$
describes our actual physical system as well as $\omega$. In this
sense, the states in $\Omega_\eps(A_1,\, \ldots ,\, A_k;\omega)$ are good
approximations for $\omega$.

Besides, it is a mathematical fact that this kind of neighbourhoods
induce in $\Sigma$ a topology, called in \cite{araki} \emph{physical}
or \emph{weak topology}, whose notion of convergence may be cast in
the following way: one says that a sequence\footnote{Strictly speaking
  one should use {\em nets} to characterise convergence in those
  topologies, instead of sequences. For the sake of clarity we will
  use the latter. See, {\em e.g.}, \cite{megginson}.} of states $\omega_n$
converges \emph{physically}, or \emph{weakly}, to $\omega$ if, for any
observable $A$, the mean values converge, \emph{i.e.}, if one has
$\omega_n(A) \longrightarrow \omega(A)$ as $n \rightarrow \infty$. As
a conclusion, $\omega$ may be completely determined by means of a
process of taking limits in the weak sense.

It is important to clarify some differences in
nomenclature: among mathematicians the topology referred to as 
\emph{weak topology} is called \emph{weak-$\ast$ topology}. The reason
relies on the fact that sets defined like
$\Omega_\eps(A_1,\, \ldots ,\, A_k;\omega)$, with $A_1,\ldots, A_k$ in a normed
vector space $\mathcal{V}$, and $\omega$ any bounded linear functional
on $\mathcal{V}$, form a basis of neighbourhoods for a locally convex
topology on the dual of $\mathcal{V}$.

Topologies can be similarly introduced in the space of observables,
providing notions of closeness between operators. We will make use of
some of then in our final sections.

In the so-called {\em norm (or uniform) topology}, two operators $A$
and $A'$ acting on a Hilbert space $\H$ are considered close if for
some prescribed $\eps>0$ one has
$\|A-A'\|:=\sup\big\|(A-A')\psi\big\|<\eps$, where the supremum is
taken over all vectors $\psi\in\H$ with $\|\psi\|=1$. This means that
$A\psi$ and $A'\psi$ differ in norm by an amount smaller than the
prescribed error $\eps$ regardless of the normalised vectors $\psi$.  In
this topology, we say that a sequence of operators $A'_n$ converges to
$A$ if for any $\|A-A'_n\|$ goes to zero when $n$ goes to infinity.

In the so-called {\em strong operator topology} in a Hilbert space
$\H$, two operators $A$ and $A'$ are considered close with respect to
a distinct finite set of normalised vectors
$\psi_j\in\H, \; j=1, \; \ldots , \; N$, and for some prescribed
$\eps >0$, if one has $\big\|(A-A')\psi_j\big\|<\eps$, for all
$j=1, \; \ldots , \; N$. This means that the vectors $A\psi_j$ and
$A'\psi_j$ differ in norm, for each $j=1, \; \ldots , \; N$, by an
amount smaller than the prescribed error $\eps$. In this topology, we
say that a sequence of operators $A'_n$ converges to $A$ if they
eventually become close, when $n\to\infty$, with respect to all finite
sets of normalised vectors $\psi_j\in\H, \; j=1, \; \ldots , \; N$,
and all $\eps >0$.

In the so-called {\em weak operator topology} in a Hilbert space $\H$,
two operators $A$ and $A'$ are considered close with respect to a
distinct set of normalised vectors
$\psi_j\in\H, \; j=1, \; \ldots , \; N$, and for some prescribed
$\eps >0$, if one has
$\big|\langle \psi_j, \; (A-A')\psi_j\rangle\big|<\eps$, for all
$j=1, \; \ldots , \; N$. This means that $A$ and $A'$ provide the same
expectation values for the vector states defined by the vectors
$\psi_j$, $j=1, \; \ldots , \; N$, up to an error smaller than the
prescribed $\eps$. In this topology, we say that a sequence of
operators $A'_n$ converges to $A$ if they eventually become close when
$n\to\infty$ with respect to all finite sets of normalised vectors
$\psi_j\in\H, \; j=1, \; \ldots , \; N$, and all $\eps >0$.

In the weak operator topology the notion of closeness between
operators is expressed in terms of their expectation values and,
therefore, is directly linked to measurable quantities.

In infinite dimensional Hilbert spaces all operator topologies defined
above differ \cite{bratteli}, leading to distinct notions of
convergence between operators, a very significant fact for the
mathematical analysis of quantum systems. For instance, convergence of
sequences of operators in the uniform operator topology implies
convergence in the strong operator topology.  Analogously,
convergence of sequences of operators in the strong operator topology
implies convergence in the weak operator topology. In both cases the
opposite statements are not generally valid.

\section{The Algebraic Approach to Quantum Systems}

The familiar Hilbert space approach, however, has limitations when
dealing with quantum systems with infinitely many degrees of freedom,
as those considered in Quantum Field Theory and Quantum Statistical
Mechanics, mainly due to important features commonly manifest in such
systems, such as superselection sectors, phase transitions and the
existence of some special states, for instance ``thermal or finite
temperature'' states, that cannot be properly described in the Hilbert
space formalism. A universal formalism that can be applied to general
quantum systems was proposed by Haag, Kastler and many others (see
{\em e.g.}, \cite{araki,haag} or, for a more recent review,
\cite{FewsterRejzner}).  We will refer to this formalism as the
algebraic approach to quantum systems. In general terms, it emphasises
the dichotomy between observables (representing physically meaningful
and measurable quantities) and states (dealing with the statistics of
measurements of physically observable quantities).

In this formalism, observables are treated as abstract associative
algebras of a certain kind (usually $C^\ast$ and/or von Neumann
algebras are considered), while states are associated to positive
normalised linear functionals on these algebras.  Within the algebraic
approach one is no longer restricted to the use of the vectors or
density operator on $\H$ in order to describe states.  Moreover, this
formalism allows the treatment of pure and mixed states in a very
general and elegant fashion, a point that will be relevant for our
purposes.

As we mentioned, it was due to the efforts of Haag and others that
\Cstar-algebras have been recognised as the relevant mathematical
objects for the universal description of observables in quantum
systems. Let us briefly describe such algebras.

Let us denote by $\mathcal{O}$ the set of observables of a physical
system. 
$\mathcal{O}$  must have a real vector space structure; moreover,
the composition of some observables must result in a new observable. 
This suggests that $\mathcal{O}$ must be contained in a
larger set that is an associative algebra.  Let us denote a minimal
associative algebra satisfying these properties by $\A$. Hence, $\A$
is a (complex) vector space endowed with associative multiplication.
We also require the existence of an
operation $\ast : \A \longrightarrow \A$, called an \emph{involution}
in $\A$, such that for all $A, \; B\in\A$ and $z\in\C$ one has:
$(A^\ast)^\ast=A$, $(A+B)^\ast = A^\ast + B^\ast$,
$(zA)^\ast=\overline{z}A^\ast$ and\footnote{These rules are not
  supposed to hold in the case of unbounded operators acting on Hilbert
  spaces. See {\em e.g.}. \cite{reed_simon}.}
$(AB)^\ast = B^\ast A^\ast$. Here $\overline{z}$ denotes the complex
conjugate of $z\in\C$.

The algebra $\mathcal{L}(\H)$, for example, possesses these properties
and the involution is related the notion of the \emph{adjoint} of a
bounded operator acting on $\H$ with respect to the scalar product on
$\H$.

The requirement of associativity in quantum systems deserves some
physical clarification. Regarding the elements of $\A$ as {\em
  operations} acting on a quantum system, the order of two successive
operations matters and the algebra of observables is not supposed to
be commutative. If we consider three successive operations, however,
one has to guarantee that the last operation does not depend on the
previous ones, what is achieved through the requirement of
associativity.

Finally, since $\A$ is an extension of the more physically relevant
set $\mathcal{O}$, we need a way to distinguish $\mathcal{O}$ within
$\A$. The involution provides a method to identify the elements of
$\mathcal{O}$ among all elements of $\A$: let $A \in \A$, if
$A \in \mathcal{O}$, then $A=A^\ast$.  One can wonder if $A^\ast=A$
implies $A\in\mathcal{O}$. In the first attempts to axiomatise quantum
mechanics it was required that any self-adjoint operator should
represent an observable, but for physical reasons this requirement was
discarded. There are many examples of self-adjoint operators not
associated to observables. If $U$ is any unitary operator acting on a
Hilbert space of physical states $\H$, then $U+U^*$ is self-adjoint,
but it may not be associated to a measurable quantity.  For example,
when $U$ represents the shift operator on a separable Hilbert space,
acting on an orthonormal base of vectors $\{\phi_n\}_{n\in\Z}$ as
$U\phi_n=\phi_{n+1}$). Another example: consider a Fermionic field
$\psi$ and take $\psi+\psi^*$, a self-adjoint operator not related to
a measurable quantity. For a more detailed discussion on the axioms
of quantum mechanics see \cite{Emch}.

The last ingredient in our construction is a norm on $\A$ which must
be compatible with the multiplication and the involution
operations. We now define:

\begin{definition}
  A \Cstar-algebra is a set $\A$ provided with a complex linear
  structure, an associative multiplication, an involution and a norm
  such that
	\begin{enumerate}[(i)]
        \item $\lambda(AB)=(\lambda A)B=A(\lambda B)$, for all
          $A,B\in\A$ and $\lambda \in \mathbb{C}$;
        \item $A(B+C)=AB+AC$, for all $A,B, C\in\A$;
        \item $\|AB\|\leq \|A\| \|B\|$, for all $A,B\in\A$;
        \item $\|A^\ast A\|=\|A\|^2$, for all $A\in\A$;
        \item $\A$, as a vector space, is complete with the norm.
	\end{enumerate}
\end{definition}

The algebra $\L(\H)$, of all bounded (continuous) operators acting on a
Hilbert space $\H$, is known to be a \Cstar-algebra.  One might ask
whether the definition above leads to anything different from the
usual description of Quantum Mechanics based in Hilbert spaces. The
answer to this question is negative and is based on the following
facts (respectively, Theorems 2.1.10 and 2.1.11A of \cite{bratteli}):

\begin{theorem}\label{th:algebra_isomorphism}
  Let $\A$ be a \Cstar-algebra. There exists a Hilbert space $\H$
  such that $\A$ is isomorphic to some self-adjoint closed subalgebra
  of $\L(\H)$.
\end{theorem}

So, essentially, a \Cstar-algebra is an abstract algebra of
operators. Furthermore:

\begin{theorem}\label{th:algebra_isomorphism_function}
  If a \Cstar-algebra $\A$ is commutative, then there exists a locally
  compact Hausdorff topological space $X$ such that $\A$ is isomorphic
  to $C_0(X)$, \emph{i.e.}, the algebra of continuous complex
  functions on $X$ that vanish at infinity\footnote{A function
    $f \in C_0(X)$ vanishes at infinite if, for any $\eps > 0$, there
    is $K \subset X$ compact such that $|f(x)| < \eps$ for any
    $x \in X \setminus K$.}.
\end{theorem}


In other words, (commutative) \Cstar-algebras can be viewed as
abstract algebras of functions. Hence the notion of states as
functionals over $\A$ is suitable for both quantum and classical
theories, as well as any other experimental theories.

Besides, on one hand a \Cstar-algebra can always be mapped to a
closed $\ast$-subalgebra of $\L(\H)$ for some suitable Hilbert space
$\H$. On the other hand, such $\H$ may be obtained by means of the
\emph{GNS construction} that we will present in Section
\ref{sec:irreducibility}.

Finally, just like the operators in $\L(\H)$, the elements of a
\Cstar-algebra possess adjoints, norms, and even spectra: given
$A \in \A$, the spectrum of $A$ is the set
$${\rm spec}(A) = \{ \lambda \in \C \; : \; A-\lambda\mathbbm{1} \text{  is not invertible} \}.$$
Indeed, one may speak about the inverse of an element $A$ in a
\Cstar-algebra since either it has an identity $\mathbbm{1}$, or we
may map $\A$ to a larger algebra $\tilde{\A}$ containing an identity,
a case where we consider the spectrum of $A$ in $\A$ to be its
spectrum as an element in $\tilde{\A}$\footnote{Concretely,
  $\tilde{\A}$ is the algebra with elements
  $(\lambda,\, A) \in \C \times \A$, involution
  $(\lambda,\, A)^\ast = (\overline{\lambda},\, A^\ast)$ and multiplication
  $(\lambda_1,\, A_1)(\lambda_2,\, A_2) = (\lambda_1\lambda_2,\, \lambda_1A_2 +
  \lambda_2A_1 + A_1A_2)$; one easily checks that $(1,\, 0)$ is an
  identity in $\tilde{\A}$, and the convenient way to map $\A$ into
  $\tilde{\A}$ is through the application $A \longmapsto (0,\, A)$.}. As
a result, we can still have the usual interpretation of the spectrum
of a self-adjoint operator as the possible outcomes of physical
experiments. Moreover, some \Cstar-algebras admit traces,
\emph{i.e.} positive functionals acting on them that generalise the
usual notion of trace of an operator \cite{bratteli,dixmier}.

A last comment worth mentioning is that if traditionally one considers
unbounded operators such as momentum or the Hamiltonian as
observables, in practice one never performs measurements capable of
observing the entire set of possible values of momenta or energy when
they are not bounded. Sensors and physical equipment have always a
bounded range within which they are suitable for making measurements, so what
is really done in a physical theory is to account for measurements of
a real observable $A$ within a certain bounded interval
$I \subset \R$, whose correspondent operator $A\chi_I\left(A\right)$
is bounded by $\rm{supp}\, |I|$. Here, $\chi_I$ is the so-called
characteristic function of $I$: for real $x$, the function
$\chi_I\left(x\right)$ equals $1$ for $x\in I$ and $0$ otherwise.  The
operator $\chi_I(A)$ is defined by the functional calculus for
self-adjoint operators $A$. An unbounded observable $A$ shall thus be
thought of as some kind of limit $A\chi_{I_n}\left(A\right)$ along an
increasing sequence of intervals $I_n\nearrow \R$ (for more details,
see the notion of affiliated operators in \cite{bratteli}).


\section{States as Functionals on \Cstar-Algebras}\label{sec:states_as_functionals}

Since we have re-elaborated our concept of observables, that of
states ought to be rediscussed too. Indeed, as we have argued in Section \ref{sec:the_notion}, a state could be thought of as some property of a
physical system that associates to each observable $A$ a probability measure describing the statistical distribution of repeated measurements of $A$ on many copies of the same system.

This can be achieved if  we associate to $A$ the number
$\omega(A)$ corresponding to the mean value of all its outcomes over the several measurements, for if one knows the averages of any observables, including those like ``the frequency of the outcome $a$ in an experiment measuring $A$'', then it is possible to reconstruct the statistical distribution for $A$. For this reason, states will be defined as functions from $\A$ into the theory's scalars, usually $\C$, with the special property that $\omega(A)$ must be a real number if $A$ is indeed a physical observable, \emph{i.e.}, if $A \in \mathcal{O} \subset \A$. 

Other reasonable properties for $\omega$ to have a physical meaning are:
\begin{itemize}
\item \textbf{Positivity.} For a positive observable $A$, one must
  have $\omega(A) \geqslant 0$, for if any possible measurement of $A$ results in
  a positive value, so must be their average. This is equivalent to
  saying that $\omega(A^* A) \geqslant 0$ for any $A \in \A$.
\item \textbf{Boundedness.} The average of a set of values cannot be
  greater than their supremum; the measurable values of an observable
  $A$ are the elements of its spectrum, which is bounded by
  $\|A\|$. It follows that $\omega(A) \leqslant \|A\|$, or, more
  shortly:
\begin{equation}\label{eq:supremum}
  \|\omega\| = \sup_{\underset{\rm observable}{A \neq 0}} \frac{|\omega(A)|}{\|A\|} \leqslant 1.
\end{equation}
For a matter of convenience, we always take the supremum over the entire $\A \setminus \{0\}$.
\item \textbf{Normalisation.} $\omega(A)$ is to be the average of
  measurements of $A$. If the algebra $\A$ has an identity
  $\mathbbm{1}$, whose only measurable value is $1$, we must impose
  $\omega(\mathbbm{1}) = 1$, consequently the supremum in
  \eqref{eq:supremum} is achieved and one gets $\| \omega \| = 1$. If
  the algebra has no identity, there will be a bounded non-decreasing
  net\footnote{\label{ft:order} An increasing net of operators is a
    net such that $A_\nu \geqslant A_\mu$ whenever
    $\nu \geqslant \mu$; the operators' order relation is defined in
    the following way: $A_\nu \geqslant A_\mu$ if $A_\nu - A_\mu$ is a
    positive operator, \emph{i.e.}, if its spectrum is included in
    $\R^+$.} $(A_\nu)_{\nu \in I}$ of operators approximating the
  identity (see Theorem 2.2.18 of \cite{bratteli}), which causes
  $\omega(A_\nu)$ to approximate $1$ and the supremum in
  \eqref{eq:supremum} to be exactly $1$. In any way, we end up with
  $\| \omega \| = 1$.
\item \textbf{Linearity.} In the usual description of Quantum Physics,
  the average of linear combinations of two observables $A$ and $B$,
  $\lambda \in \R$, is given by the combination of their averages:
  $$
  \< A + \lambda B\> = \<A\> + \lambda \<B\>,
  $$
  irrespective to whether $A$ end $B$ are compatible observables or
  not, {\em i.e.}, whether the corresponding operators $A$ and $B$
  commute or not.  Hence, we must have
  $\omega(A+\lambda B) = \omega(A) + \lambda \omega(B)$ at least for
  $A$ and $B$ being themselves observables and $\lambda \in \R$, and
  so being $A+\lambda B$. If one of these is not an observable (for
  instance if $A = i \mathbbm{1}$, $B = -i \mathbbm{1}$ and
  $\lambda = 1$), it is not physically clear what their linear
  combination should mean, nor even which interpretation one should 
  give to $\omega(A)$. Therefore, imposing linearity on $\omega$ with
  respect to $A$ for any $A \in \A$ is an arbitrary choice, so as to
  end up with a linear theory; linearity is apparently not a
  physical requirement unless for the restriction of $\omega$ to real
  linear combinations of elements of $\mathcal{O}$. Nonetheless, the
  theory we obtain doing so seems to be a good description of what is
  actually seen in the laboratory experiments.
\end{itemize}

We are thus led to:
\begin{definition}\label{def:state}
  Given a \Cstar-algebra $\A$, $\omega$ is said to be a state over
  $\A$ if it is a bounded positive linear functional with
  $\|\omega\| = \sup_{\underset{A \neq 0}{A \in \A}}
  \frac{|\omega(A)|}{\|A\|} = 1$.
\end{definition}

   %

\begin{remark}\label{rem:ppc}
  Supposing that a \Cstar-algebra $\A$ has an identity
  $\mathbbm{1} \in \A$, it is possible to show (see {\em e.g.},
  \cite{bratteli,Davidson,Arveson}) that a linear functional $\omega$
  on $\A$ is bounded with $\| \omega \| = \omega(\mathbbm{1})$ if and
  only if it is positive. This close connection between positivity and
  boundedness allows us to require in Definition \ref{def:state} that
  a state be merely a positive linear functional satisfying
  $\omega(\mathbbm{1})=1$. If the algebra has no identity, then it has
  at least a non-decreasing net $\left(A_\nu\right)_{\nu \in I}$
  approximating it (Theorem 2.2.18 of \cite{bratteli}), and this
  remark holds under the form
  $\|\omega\| = \omega\left(\sup_{\nu \in I} A_\nu\right)$.
\end{remark}

The first evident virtue of this new notion of state is
epistemological, as it applies not only to Quantum Physics, but also
for any experimental science whose systems are in states about which
we have information exclusively through series of measurements, so we
can only determine them by the statistical profile of the data we
gather. This holds regardless of any predefined deterministic notion
of state, and regardless of having or having not complete knowledge
about the system that one might obtain from a totally accurate
measurement.

An obvious example of such an experimental theory is Classical Physics
itself, and in Section \ref{sec:states_experiment} we show how this
rigorous notion of state fits perfectly the classical situation from
the laboratory's point of view (recall Theorem
\ref{th:algebra_isomorphism_function} above).

Another of this concept's advantages is that it does encompass more
quantum states than previously; in Section
\ref{sec:example_not_vector} we exhibit a state that cannot be written
neither as a vector nor as a density operator. In this case, however,
the reader should pay attention to the point that this state will be
defined through a process of taking limits, what is not fortuitous. In
fact, any state on a \Cstar-algebra $\A$ may be approximated by a
sequence of states that correspond to density operators when $\A$ is
realised as a concrete operator algebra $\mathcal{L}(\H)$ (see
Theorems \ref{th:algebra_isomorphism},
\ref{th:normal_density_operator} and \ref{th:normal_is_dense}).

Finally. we notice that if $\omega_1$ and $ \omega_2$ are two states in $\A$
and $\lambda \in [0, \; 1]$, it can be easily verified form the
definition that the convex combination
$\lambda \omega_1 + (1-\lambda)\omega_2$ is also a state in $\A$.
This remark will be essential for the notion of pure and mixed states below.


\subsection{Pure and mixed states}\label{sec:Pureandmixedstates}

Now we come to the central notions in this discussion.  As informally
described above, \emph{pure states} are those that cannot be written
as a convex combination of other states. This concept may be
formalised for \Cstar-algebras in the following way:

Let be $\A$ a \Cstar-algebra, and $\omega$ a state on it. $\omega$ is
said to be pure if, given a scalar $\lambda \in [0, \; 1]$ and states
$\omega_1 \neq \omega_2$ on $\A$ such that
$$\omega = \lambda \omega_1 + (1-\lambda)\omega_2,$$
then necessarily $\lambda = 0$ or $\lambda = 1$. Otherwise, $\omega$
is said to be a mixture of $\omega_1$ and $\omega_2$.

In Section \ref{sec:irreducibility} we will present an important
criterion for characterising a state $\omega$ on $\A$ as pure. It will
be done by means of evaluating the reducibility of the representation
of $\A$ into a specific algebra of operators on a Hilbert space
especially constructed for $\A$ and $\omega$, the so-called \emph{GNS
  representation}. This characterisation is very important, because it
makes clear the point that a mixed state is a sort of composit system
where each component can be analysed by an independent algebra of
observables, namely, by an irreducible representation in the
\emph{GNS} Hilbert space of the abstract observable algebra.

\subsection{Pure and vector states}\label{sec:pure_and_vector}

When the \Cstar-algebra is a subalgebra of $\L(\H)$ for some Hilbert
space $\H$, there is a natural way of defining states on it, which is
to take a $\Psi \in \H$ with norm $\|\Psi\| = 1$ and put
$$\omega_\Psi(A) = \<\Psi,\, A\Psi\>$$
for any $A \in \A$. Proving that this $\omega_\Psi$ is a true state in
the sense of Definition \ref{def:state} is straightforward: obviously
it is linear and positive (for
$\omega_\Psi(A^* A) = \|A\Psi\|^2 \geqslant 0$), so by Remark
\ref{rem:ppc}, $\|\omega_\Psi\| = \|\Psi\|^2 = 1$.

 A state $\omega$ on $\A \subset \mathcal{L}(\H)$ such that there exists a unit vector $\Psi$ in $\H$ satisfying
\[\omega(A) = \<\Psi,\, A\Psi\>,\]
for any $A \in \A$ is said to be a vector state.

If the \Cstar-algebra considered is the algebra of all continuous
operators acting on a Hilbert space, one has the following important
statement:

\begin{theorem}\label{th:vector_is_pure}
 Any vector state on the \Cstar-algebra $\L(\H)$ is pure.
\end{theorem}

This theorem is a particular case of Theorem
\ref{th:vector_is_pure-generalisation}, whose proof is given below.

This theorem may explain why some Quantum Mechanics textbooks present
the concepts of pure and of vector states as equivalent. It is
important to emphasise, however, that these concepts are not
equivalent, in general, and one may find situations, even physically
relevant ones, where certain pure states are not vector states and
certain vector states are not pure.

Indeed, in Section \ref{sec:example_not_vector} we exhibit an example
of a pure state that is not a vector state. Moreover, in the presence
of superselection rules, one can see on physical grounds that not all
vector states are pure.

As an example, consider the simple case where we have two
superselection sectors corresponding to some conserved ``charge''
assuming two distinct values. The physical Hilbert space $\H$ is a
direct sum of two mutually orthogonal subspaces $\H=\H_1\oplus \H_2$,
but the algebra of observables cannot be the whole $\L(\H)$, since there
are operators in $\L(\H)$ that map $\H_1$ into $\H_2$ (and
vice-versa), violating the superselection rule. Thus, the algebra of
observables has to be a subalgebra of $\L(\H_1)\oplus \L(\H_2)$. Let
us assume for simplicity that the algebra of observables coincides
with $\L(\H_1)\oplus \L(\H_2)$.
Take a vector in $\H$ in the form
$\Psi=\big(a_1 \psi_1\big)\oplus \big(a_2 \psi_2\big)$, where
$\psi_1\in\H_1$ and $\psi_2\in\H_2$ are normalised vectors ({\em
  i.e.}, $\|\psi_1\|_{\H_1}=\|\psi_2\|_{\H_2}=1$) and where $a_1$ and
$a_2$ are non-zero complex numbers with $|a_1|^2+|a_2|^2=1$ . Then,
$\Psi$ is also normalised and, for any observable $A=A_1\oplus A_2$,
one has, 
$$
\omega_\Psi(A) 
\; = \;
|a_1|^2\omega_1(A) + |a_2|^2\omega_2(A) \; , 
$$
where
$\omega_1(A)=\omega_1\big(A_1\oplus A_2\big):=\omega_{\psi_1}(A_1)$
and
$\omega_2(A)=\omega_2\big(A_1\oplus A_2\big):=\omega_{\psi_2}(A_2)$
are two states on $\L(\H_1)\oplus \L(\H_2)$.  Thus, the vector state
$\omega_\Psi$ is not a pure state on $\L(\H_1)\oplus \L(\H_2)$, but a
mixture of $\omega_1$ and $\omega_2$. The relation between purity and
indecomposability of the algebra will be further discussed in Section
\ref{sec:irreducibility}.

Although being a pure state does not imply being a vector state in the
general picture, as stressed above, in some important cases this
happens to be true; together with Theorem \ref{th:vector_is_pure},
this means that in these situations both concepts are indeed
equivalent. In Section \ref{sec:example_pure_is_vector}, for instance,
we show that for a \Cstar-algebra composed by all the compact
operators on some Hilbert space, all pure states are vector ones.

We shall need a technical generalisation of  Theorem \ref{th:vector_is_pure}:

\begin{theorem}\label{th:vector_is_pure-generalisation}
  Let $\H$ be some Hilbert space and $\A\subset\L(\H)$ be a
  $C^\ast$-subalgebra of $\L(\H)$.  Consider a normalised vector
  $\Phi\in\H$. If the orthogonal projection on the subspace generated
  by $\Phi$ is an element of $\A$, then the vector state
  $\omega_{\Phi}$ on $\A$ is pure.
\end{theorem}

In the case when $\A=\L(\H)$, this implies Theorem
\ref{th:vector_is_pure}, above, since in this case all orthogonal
projectors on the unidimensional subspaces generated by the vectors of
$\H$ belong to $\A$.

\begin{proof}
  We follow closely the joint proof of Theorem 2.8 and of Lemma 2.9 in \cite{araki}.
  For simplicity, let us assume that $\A$ contains a unit $\UM$.

  By contradiction, let us assume that $\omega_\Phi$ is a mixed state on $\A$. Then, there are
  $\lambda\in (0, \; 1)$ and two distinct states $\omega_1$ and $\omega_2$ on $\A$ such that
  \begin{equation}\label{eq:assumption}
    \<\Phi,\, A\Phi\>\; = \; \lambda \omega_1(A)+(1-\lambda)\omega_2(A)
   \end{equation}
   for all $A\in\A$. Let $E$ be the orthogonal projector on the
   one-dimensional subspace generated by $\Phi$. By assumption
   $E\in\A$ and $\UM - E \in\A$, and we may write
  $$
   0\; = \; \<\Phi,\, (\UM-E)\Phi\> \; = \; \lambda  \omega_1(\UM-E)+(1-\lambda)\omega_2(\UM-E)\; .
   $$
   It follows from this that
   $\omega_1(\UM-E)=\omega_2(\UM-E)=0$. Using the Cauchy-Schwartz
   inequality, one has for all $B\in\A$,
   $$
   \Big|\omega_a\big((\UM-E)B\big)\Big|^2
   \; \leq \;
   \omega_a\big((\UM-E) (\UM-E)^*\big)
   \omega_a\big(B^*B\big)
   \; = \;
   \omega_a(\UM-E) \omega_a\big(BB^*\big)
   \; = \; 0\;,
   $$ 
   what implies $\omega_a\big((\UM-E)B\big)=0$ for both $a=1, \;
   2$. Analogously, one has $\omega_a\big(B(\UM-E)\big)=0$, for both
   $a=1, \; 2$.  Since
   $$
     B \; = \; \big(E + (\UM-E)\big)B\big(E + (\UM-E)\big)
     \; = \;
     EBE + (\UM-E) BE + EB (\UM-E) + (\UM-E)B(\UM-E)
     \;,
     $$
     one has $\omega_a(B)=\omega_a(E B E)$  for both    $a=1, \; 2$. 
     Now, for any $\Psi\in\H$ one has $E\Psi=\<\Phi, \Psi\>\Phi$ and,
     hence,
     $$
     E B E\Psi \, = \, \<\Phi,\, \Psi\>E B\Phi \, = \,
     \<\Phi,\, \Psi\>\<\Phi, \; B\Phi\>\Phi
     \, = \,
     \<\Phi, \; B\Phi\> \big( \<\Phi, \Psi\>\Phi \big)
     \, = \,  \<\Phi, \; B\Phi\> E\Psi\;,
     $$
     which implies $E B E=\<\Phi, \; B\Phi\> E$. It follows that
     $\omega_a(B)=\omega_a(E B E)=\<\Phi, \;
     B\Phi\>\omega_a(E)$. Taking, in particular, $B=\UM$, this says
     that $1=\omega_a(E)$ for both $a=1, \; 2$. Hence,
     $\omega_a(B)=\<\Phi, \; B\Phi\>$ for each $a=1, \; 2$ and for all
     $B\in\A$, what contradicts \eqref{eq:assumption} with $\omega_1$
     and $\omega_2$ being distinct. This completes the proof.

     The case when $\A$ does not contain a unit can be treated
     similarly by using the so-called approximants of the identity (see,
     {\em e.g.}, \cite{bratteli}).
\end{proof}

\section{Examples}\label{sec:examples}


\subsection{Mixtures of vector states in Quantum Mechanics}\label{sec:example_mixture}
Let us consider the simple case where our \Cstar-algebra is just
$\L(\H)$, for $\H$ a Hilbert space representing a non-interacting
two-level system of particles, namely $\H = \C^2$, where level $1$ is
represented by
$\mathbf{e}_1 = \left( \begin{smallmatrix} 1 \\ 0\end{smallmatrix} \right)$ and
level $2$ by
$\mathbf{e}_2 = \left( \begin{smallmatrix} 0 \\ 1 \end{smallmatrix} \right)$. In
this section we are going to show that in general there is no vector state
$\Psi \in \C^2$ that could represent a mixture of states $\Phi_1$ and
$\Phi_2$ with respective statistical weights $p_1$ and $p_2$,
\emph{i.e.}, that it may be that no $\Psi$ is such that
$\<\Psi,\, A\Psi\> = p_1 \<\Phi_1,\, A\Phi_1\> + p_2 \<\Phi_2,\, A\Phi_2\>$ for
all observables, that turn out to be self-adjoint operators
$A \in \L(\H)$ (in this case, Hermitian $2\times 2$ complex matrices).

Take a mixture of $p_1 = \frac{1}{3}$ particles in state
$\Phi_1 = \mathbf{e}_1$ and $p_2 = \frac{2}{3}$ in
$\Phi_2 = \mathbf{e}_2$; put
$\Psi = \left( \begin{smallmatrix} \psi_1 \\ \psi_2 \end{smallmatrix} \right)$,
with $\psi_1,\psi_2 \in \C$. In order that the mean values of an
observable
$A = \left( \begin{smallmatrix} a_1 & 0 \\ 0 & a_2 \end{smallmatrix} \right)$
coincide when measured for $\Psi$ and for the mixture, it is necessary
that $|\psi_1|^2 = \frac{1}{3}$ and $|\psi_2|^2 = \frac{2}{3}$; we
could thus choose $\psi_2 = \sqrt{\frac{2}{3}}$ and
$\psi_1 = e^{i\gamma} \sqrt{\frac{1}{3}}$, with some phase
$\gamma \in \R$ that can be fixed if we take an observable
$B = \left( \begin{smallmatrix} 0 & b \\ b & 0\end{smallmatrix} \right)$ and
impose again that the mean values of $B$ must coincide in both cases:
we find out that $\gamma = \frac{\pi}{2}$. As a conclusion, one has
$\Psi = \frac{1}{\sqrt{3}} \left( \begin{smallmatrix} i \\
    \sqrt{2}\end{smallmatrix} \right)$.

Nonetheless, take the observable
$C = \left( \begin{smallmatrix} 0 & i \\ -i & 0 \end{smallmatrix} \right)$. The
mean value of $C$ for the mixture is $0$, whereas
$\<\Psi,C\Psi\> = 2\sqrt{2}$. Therefore, no state $\Psi \in \C^2$ may
represent the mixture consisting of $\frac{1}{3}$ of particles in the
state $\textbf{e}_1$ and $\frac{2}{3}$ of them in $\textbf{e}_2$.

In fact, the algebra of observables of a two-level system coincides
with the set of all complex $2\times 2$ matrices.  Moreover, we know
(see Sec.\ \ref{sec:normal_states}) that a general state for this
algebra is of the form $\omega_\rho(A)=\Tr(\rho A)$, where $\rho$, the
so-called {\em density matrix }, is a self-adjoint and positive matrix
such that $\Tr\rho=1$. Since $\rho$ is a self-adjoint operator, it can
be written in the form
$\rho=\frac{1}{2} \big( a_0\UM+\vec{a}\cdot \vec{\sigma} \big)$, where
$\vec{a}\cdot \vec{\sigma}$ is a notation for
$a_1\sigma_1+a_2\sigma_2+a_3\sigma_3$, with $a_k$,
$k=0, \, \ldots , \, 3$, being real numbers and $\sigma_l$,
$l=1, \, 2, \, 3$, being the Pauli matrices\footnote{The
  matrices $\UM$, $\sigma_1$, $\sigma_2$ and $\sigma_3$ are a basis in
  the real space of the $2\times 2$ self-adjoint matrices.}. The
condition $\Tr\rho=1$ implies $a_0=1$. In this case, the eigenvalues
of $\rho$ are $ \varrho_1 = \frac{1+\|\vec{a}\|}{2} $ and
$ \varrho_2 = \frac{1-\|\vec{a}\|}{2} $ and, hence, the condition of
$\rho$ having strictly positive eigenvalues is $\|\vec{a}\| < 1$. For
$\|\vec{a}\| = 1$ one has $\varrho_1=1$ and $\varrho_2=0$, implying
that the matrix $\rho$ is an orthogonal projector.  Hence, we can
associate the space of states for a two-level system with a closed
unit sphere centred at the origin in three dimensional space, the
so-called Bloch sphere, with the pure states being those on the surface
of the sphere (corresponding to $\|\vec{a}\|=1$) and with the mixed
states being inside the sphere (corresponding to $\|\vec{a}\|<1$).


\subsection{Experiment-determined states in Classical Physics}\label{sec:states_experiment}
Classically, the trajectory of a particle submitted to a smooth
potential remains completely determined given its position and
momentum at a certain instant; besides, any other physical quantity
may be expressed in terms of these data, like the energy and the
angular momenta. Thus, the phase space becomes a natural environment for
describing the states of classical systems, whose observables are
scalar functions defined over this phase space.

Effectively, although in the classical theory the particle's states
are understood as the points of the phase space, even the classical
experiments alone cannot determine them with total accuracy: an
experiment only furnishes a set of data permitting one to depict a
(usually continuous) probability distribution with a mean value and a
standard deviation, which allows us to calculate the probability of
finding a particle within some region in the phase space, though not
to assert that the particle will be precisely in one point or
another. This does not imply that points are not states; they indeed
are, since in Classical Mechanics we admit the particle to have fully
determined values of position and momentum, could we measure them or
not. However, this rather indicates the need for more states with
lesser localisation properties in order to take into account the fact
that we may also have access to reality through experiments subjected
to statistical errors.

Actually, it is possible to realise a commutative \Cstar-algebra as
an algebra of $C_0$ functions on some phase space (as reads Theorem
\ref{th:algebra_isomorphism_function}), the states of which being linear
positive and normalised functionals on such continuous functions. By
the Riesz-Markov representation theorem (see {\em e.g.},
\cite{rudin}), a linear, bounded and positive functional $\omega$
acting on such a function $f$ can be written as an integral over a
probability measure:
$$
 \omega(f) = \int f d\mu_\omega,
$$    
where $\mu_\omega$ is a positive measure over the phase space; the
normalisation condition gives, moreover, $\int d\mu_{\omega} = 1$, so
$\mu_\omega$ may be interpreted as a probability measure whose events
are points and (open) regions of the phase space, just as one would need in
order to represent Classical Physics as a theory seriously committed
to experimental results and their intrinsic uncertainty.

Of course this can only be done if $\A$ is commutative; however, since
commutativity basically means that the observables are compatible,
\emph{i.e.}, can be measured at the same time, and in Classical
Physics this is always the case (for it is supposed that a measurement
does not modify a system's state), such a restriction on the
\Cstar-algebras is but very natural.

Notice, though, that Dirac delta measures centred at a single point
in phase space of a classical mechanical system, are genuine
probability measures, so states of completely defined position and
momentum are not excluded; on the contrary, these are precisely the
classical theory's pure states. Consequently, this framework does not
exclude the possibility of a classical particle to have a definite
position and momentum. The novelty here is that our inability of
knowing them is not excluded as well, inasmuch as states are allowed
that correspond to the actual results of our experiments.

In this context, it is also relevant to remark that, in Classical
Mechanics, the variance of any observable ({\em i.e.}, a measurable
function in phase space) in a pure state vanishes identically, since
pure states are expressed by Dirac delta measures centred at a point in
phase space.  Generally, this characteristic is not shared by pure states
in quantum systems and this is one of the most relevant distinctions
between these theories.

\section{Pure Quantum States and Their Classical Limits}\label{sec:quantum_classical}

One of the physically interesting questions in this context is about
what happens with the purity of a given state of a quantum system
when the classical limit is taken. In this section we will address this
question. Naively one could believe that the purity is preserved but,
as we will now discuss, there are some interesting situations where a
pure quantum state is transformed into a classical mixed one.

As we have seen above, the pure states in Classical Physics are the
one-particle well-defined position and momentum states, namely the
measures $\delta_{x_0} \otimes \delta_{\xi_0}$ charging points
$(x_0,\xi_0)$ of the phase space (for simplicity, we assume it to be
$\R^2$). In the quantum picture for a system described by observables
in $\A = \L(\H)$, with $\H$ a Hilbert space, given an adequate normalised
vector state $\psi \in \H$ one can consider the following
$\hbar$-dependent family of wave-packets:
\begin{equation}\label{eq:wave_packet}
\psi^\hbar (x) = \frac{1}{\hbar^{\frac{1}{4}}} \psi\left(\frac{x-x_0}{\sqrt{\hbar}}\right) e^{\frac{i}{\hbar}x \cdot \xi_0}, \;\; {\rm for} \; \hbar > 0 ;
\end{equation}
it is easy to see that they satisfy
$$\lim_{\hbar \rightarrow 0} \< \psi^\hbar, \hat{H}^\hbar \psi^\hbar \>_\H = E(x_0,\xi_0),$$
where $\hat{H}^\hbar$ is the usual Hamiltonian operator (with some not
too increasing potential, say $|V(x)| \lesssim x^2$), and $E$ is the
corresponding classical energy:
$$\hat{H}^\hbar = -\frac{\hbar^2}{2}\Delta + V \qquad {\rm and} \qquad E = \frac{1}{2}\xi^2 + V(x).$$

In order to show this fact, one merely notes that $\hat{H}^\hbar$ can
be given by the integral formula
\begin{equation}\label{eq:kuetbhgiu4yb}
 \left(\hat{H}^\hbar \psi\right) (x)
 = \frac{1}{2\pi\hbar} \int_{\R_\xi} \int_{\R_y}
 e^{\frac{i}{\hbar}\xi \cdot(x-y)}
 E\left(\frac{x+y}{2},\xi\right) \psi(y) \, dy \, d\xi.
\end{equation}

Notice that, so far, we ought to take $\psi$ such that the integral in \eqref{eq:kuetbhgiu4yb}
converges, which can be assured if we take $\psi$ with sufficient decay in frequency.
However, this can be extended if we consider a function $a \in C_0^\infty(\R^2)$, so as the operator
 $\op_\hbar (a)$, defined by
$$
\big( \op_\hbar(a) \psi\big) (x)
= \frac{1}{2\pi\hbar} \int_{\R_\xi} \int_{\R_y} e
^{\frac{i}{\hbar}\xi \cdot(x-y)} a\left(\frac{x+y}{2},\xi\right) \psi(y) \, dy \, d\xi,
$$
is bounded. Hence, $\op_h(a) \in \A$ for all $\hbar > 0$. Moreover, a few
integral calculations show that:
$$\lim_{\hbar \rightarrow 0} \< \psi^\hbar, \op_\hbar(a) \psi^\hbar \>_\H = a(x_0,\xi_0).$$
The operator $\op_\hbar (a)$ is known in the literature as the
\emph{Weyl quantisation of the symbol} $a$ (see for instance
\cite{dimassi_sjostrand,zworski}, or \cite{chabu_tese} and the
references quoted therein).

With more generality, a probability measure $\mu$ on $\R^2$ is said to be the
\emph{Wigner} or \emph{semiclassical
  measure}\cite{chabu_tese,patrick_gerard,paul} associated to a
normalised family of vectors $(\psi^\hbar)_{\hbar > 0} \subset \H$ if,
for any $a \in C_0^\infty(\R^2)$, one has:
$$
\lim_{\hbar \rightarrow 0}\< \psi^\hbar, \op_\hbar(a) \psi^\hbar \>_\H
\; = \; \int_{\R^2} a(x,\xi) \mu(dx,d\xi).
$$
Within our framework of understanding states as linear
functionals over a \Cstar-algebra, we can write down formally:
$$\mu = {\rm sc}\lim_{\hbar} \omega^\hbar,$$
where $\omega^\hbar$ are states on the $C^\ast$-subalgebra
$\Upsilon \subset \A$ generated by the Weyl quantisation of symbols in
$C_0^\infty(\R^2)$, given by
$\omega^\hbar\left(A\right) = \< \psi^\hbar, A\psi^\hbar \>_\H$ for
$A \in \Upsilon$. Since $\mu(\R^2) = 1$, one can understand $\mu$ as a
classical state on a commutative \Cstar-algebra describing the
classical observables (considered as the elements of $C_0^\infty(\R^2)$), 
as in Section \ref{sec:states_experiment}.

$\Upsilon$ is indeed an algebra, for $\op_\hbar(a) \op_\hbar(b) = \op_\hbar(a\ast_\hbar b)$, with
$$
\left(a \ast_\hbar b\right)(x,\xi) = \left. e^{\frac{i\hbar}{2}
    (\partial_\xi\partial_y - \partial_x\partial_\eta)
  }\left(a(x,\xi)b(y,\eta)\right) \right|_{y = x,\,\eta =\xi};
$$
see Theorem 4.11 in \cite{zworski}. The product defined by
$\ast_\hbar$, often called the Moyal product, or Weyl–Groenewold
product, is non-commutative but is associative. The other properties
required for being a \Cstar-algebra come either trivially from the
quantisation formula (linear structure, involution) and the fact that
we are inside an algebra of bounded operators (Banach and $C^\ast$
properties), or from the definition of a generated algebra (as a
completion within the operator norm).

In order to fully characterise the Wigner measures associated to a
bounded family of vector states, it suffices to consider observables
in $\Upsilon$, since, due to the dominated convergence theorem, it is
sufficient to test a measure against functions in $C_0^\infty(\R^2)$
for entirely depicting it over $\R^2$.

A natural question that arises at this point is whether the
semiclassical limit of a family of pure states will always be pure or
not. The answer was given in \cite{chabu}, where the author showed
that a family of wave-packets not much different from
\eqref{eq:wave_packet} may split into the combination of two Dirac
delta measures on the phase space. Even better: when regarding a
time-dependent situation where the quantum states evolve according to
a dynamical equation like Schrödinger's, and their associate Wigner
measures evolve correspondingly, it is possible to have an initial state of
quantum pure states concentrating to a classical pure state, and
keeping like that for a while, but then degenerating into a classical
mixture afterwards.

More precisely, take the conical potential $V(x) = - |x|$. By the
Kato-Rellich theorem \cite{reed_simon-2}, it is known that the Schrödinger
propagator $e^{-\frac{it}{\hbar}\hat{H}}$ is unitary for any
$t \in \R$, thus, given any initial vector state $\psi_0 \in \H$, its
evolution for a time $t$, which is
$\psi_t = e^{-\frac{it}{\hbar}\hat{H}} \psi_0$, will be well and
uniquely defined. According to Theorem 1.12 of \cite{chabu}, the
semiclassical measures $\mu^1_t$ and $\mu^2_t$ associated to the
evolution of the initial data
$$\Psi^{\hbar,1}_0(x) = \frac{1}{\hbar^{\frac{1}{4}}} \Psi^1\left(\frac{x}{\sqrt{\hbar}}\right) \quad {\rm and} \quad \Psi^{\hbar,2}_0(x) = \frac{1}{\hbar^{\frac{1}{4}}} \Psi^2\left(\frac{x}{\sqrt{\hbar}}\right)e^{-i \hbar^{\beta -1} x}$$
will be, under proper choices of parameters\footnote{With
  $0 < \beta < \frac{1}{10}$, $\Psi^1, \Psi^2 \in C_0^\infty(\R)$ and
  $\Psi^1$ supported on $x > 0$.}, for $t \leqslant 0$:
$$\mu_t^1(x,\xi) = \mu_t^2(x,\xi) = \delta\left(x-\frac{t^2}{2}\right) \otimes \delta\left(\xi + t\right);$$
nevertheless, for $t > 0$, we will have:
$$\mu_t^1(x,\xi) = \delta\left(x-\frac{t^2}{2}\right) \otimes \delta\left(\xi - t\right),$$
whereas
$$\mu_t^2(x,\xi) = \delta\left(x+\frac{t^2}{2}\right) \otimes \delta\left(\xi + t\right).$$

In pictures, the particle $\mu^1$ follows the path in Figure
\ref{fig:paths}(a), and the particle $\mu^2$ moves as in Figure
\ref{fig:paths}(b).

\begin{figure}[htpb!]\center\footnotesize
\begin{multicols}{2}
\includegraphics[width=0.3\textwidth]{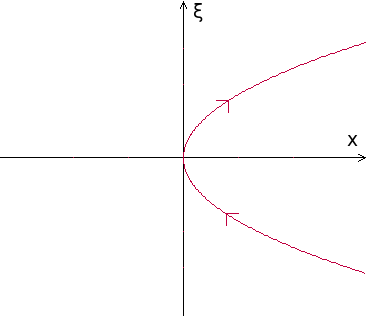} \text{(a) Trajectory of $\mu^1$.} \\
\includegraphics[width=0.3\textwidth]{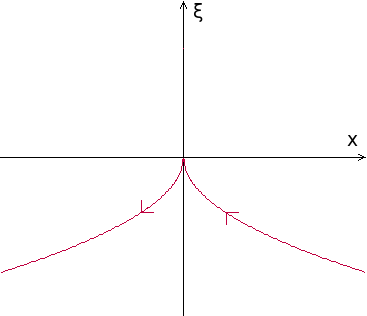} \text{(b) Trajectory of $\mu^2$.} \\
\end{multicols}
\caption{\footnotesize Trajectories followed by two different particles, coinciding for $t \leqslant 0$, but then diverging for $t > 0$. Source: \cite{chabu}.}
\label{fig:paths}
\end{figure}

Now, let us put
$\Psi^\hbar_0 = p_1 \Psi^{\hbar,1}_0 + p_2 \Psi^{\hbar,2}_0$, with
$p_1^2 + p_2^2 = 1$ (suppose
$\rm{supp}\, \Psi^1 \cap \,\rm{supp}\, \Psi^2 = \emptyset$ so we do
not need to bother about any other normalisation constants). Denoting
by $\mu_t$ the semiclassical measure related to the family
$\Psi_t^\hbar = e^{-\frac{it}{\hbar}\hat{H}}\Psi_0^\hbar$, it is
possible to show (see Proposition 3.1 and its proof in \cite{chabu})
that $\mu_t = p_1^2 \mu^1_t + p_2^2 \mu_t^2$, or, in other words,
$$
\mu_t(x,\xi) =
\delta\left(x-\frac{t^2}{2}\right) \otimes \delta\left(\xi + t\right)
\;\; {\rm for} \; t \leqslant 0,
$$
a classical pure state, and
$$
\mu_t(x,\xi) =
p_1^2 \delta\left(x-\frac{t^2}{2}\right) \otimes \delta\left(\xi - t\right)
+ p_2^2 \delta\left(x+\frac{t^2}{2}\right) \otimes \delta\left(\xi + t\right)
\;\; {\rm for} \; t > 0,
$$
a convex combination of two pure classical states, although
$\Psi_t^\hbar$ keeps being a vector, hence a pure quantum state.



\section{Pure States and Irreducibility}\label{sec:irreducibility}

We will now turn back to our general analysis and consider the
important relation between purity and irreducibility of certain
representations of the algebra of observables in quantum systems.
This discussion is of particular relevance for the treatment of
superselection sectors in Quantum Systems with infinitely many degrees
of freedom, as in Quantum Field Theory or Quantum Statistical
Mechanics.

It is possible to determine if a state $\omega$ on a \Cstar-algebra
$\A$ is pure or mixed by analysing the reducibility of a certain
representation of $\A$ into the space of bounded operators on a
Hilbert space suitably constructed from $\omega$ and $\A$. This is
achieved by means of the \emph{GNS construction} (named after Gelfand,
Naimark and Segal), which we sketch below.

To begin with, define:
\[
  \mathcal{N}_{\omega}=\left\{A \in \A ,\, \omega(A^{\ast}A)=0\right\}.
\]
$\mathcal{N}_\omega$ is a vector subspace of $\A$, as it may be
verified by the Cauchy-Schwartz inequality (see {\em e.g.},
\cite{bratteli}):
\begin{equation}\label{eq:cauchy_schwartz}
\left|\omega\left(A^* B)\right)\right|^2 \leqslant \omega(A^*A)\omega(B^*B),
\end{equation}
valid for any $A,B \in \A$ and positive linear functionals
$\omega$. Thus, one obtains a well-defined inner product
$\<\,\cdot\,,\,\cdot\,\>_\omega$ on the quotient space
$\A/\mathcal{N}_\omega$
of equivalence classes
\[
  \psi_{A}=\left\{\hat{A} \, : \, \hat{A}=A+I,\, I\in\mathcal{N}_{\omega}\right\}
  \;
\]
by posing $\<\psi_A,\psi_B\>_\omega = \omega(A^* B)$ (its independence
with respect to the classes' representatives can be verified once again
by means of the Cauchy-Schwartz inequality).

The canonical completion of $\A/\mathcal{N}_\omega$
with respect to the inner product $\<\,\cdot\,,\,\cdot\,\>_\omega$,
denoted $\H_\omega$, is called the representation space of $\A$ for
the state $\omega$.

Now, let us remark that $\mathcal{N}_\omega$ is also a left ideal of
$\A$\footnote{\label{ft:formula} For any $A,B \in \A$, one has
  $\omega((AB)^* (AB)) \leqslant \|A\|^2 \omega(B^*B)$ (see
  \cite{bratteli}), so clearly $AB \in \mathcal{N}_\omega$ whenever
  $B \in \mathcal{N}_\omega$.}, so the linear operator $\pi_\omega(A)$
acting on the dense subspace $\A/\mathcal{N}_\omega$
of $\H_\omega$ as $\pi_\omega(A) \psi_B = \psi_{AB}$ is well-defined;
moreover,
$$\|\pi_\omega(A)\psi_B\|^2 = \omega(B^* A^* AB) \leqslant \|A\|^2 \omega(B^*B) = \|A\|^2\|\psi_B\|^2 $$
(where we have used the same inequality as in Footnote
\ref{ft:formula}), showing that $\pi_\omega(A)$ is bounded and,
therefore, may be continuously extended to the whole
$\H_\omega$. Further, it is easy to see that, for any $A_1,\, A_2 \in \A$
and $\lambda \in \C$, we have
$\pi_{\omega}(A_{1})\pi_{\omega}(A_{2})=\pi_{\omega}(A_{1}A_{2})$,
$\pi_\omega(A_1+\lambda A_2) = \pi_\omega{A_1} + \lambda
\pi_\omega(A_2)$, and
$\pi_\omega\left(A_1^*\right) = \pi_\omega(A_1)^*$. It is sufficient
to check these claims for vectors of the form $\psi_B$ with
$B \in \A$, since they are dense in $\H_\omega$; so, given any
$B \in \A$:
$$\pi_{\omega}(A_{1})\pi_{\omega}(A_{2})\psi_B =\psi_{A_{1}A_{2}B} = \pi_\omega(A_1A_2)\psi_B,$$
now, using the usual vector space operations for the quotient space,
$$
\pi_{\omega}(A_{1} + \lambda A_2)\psi_B =\psi_{(A_{1}+ \lambda A_{2})B}
= \psi_{A_1B} + \lambda \psi_{A_2B} = (\pi_\omega(A_1) + \lambda \pi_\omega(A_2))\psi_B,
$$
and last
$$
\<\pi_\omega(B^*)\Psi_{A_1},\Psi_{A_2}\>_\omega = \omega((B^*A_1)^*A_2)
= \omega(A_1^*BA_2) = \<\pi_\omega(B)^*\Psi_{A_1},\Psi_{A_2}\>_\omega,
$$
leading us to the conclusion that
$\pi_\omega : \A \longrightarrow \L(\H_\omega)$ is actually a
representation of the \Cstar-algebra $\A$ in the Hilbert space $\H_\omega$.

Finally, put $\Omega_{\omega} = \psi_{\mathbbm{1}}$ if $\A$ has an identity\footnote{If $\A$ does not have an identity, similar results can
  be obtained using nets converging to a unity. See, {\em e.g.,} \cite{bratteli}.}.
%
Not only will we have that the set
$$\left\lbrace \, \pi_{\omega}(A)\Omega_{\omega}\,:\, A\in\A \, \right\rbrace$$
is the dense in $\H_\omega$\footnote{When this happens,
  $\Omega_\omega$ is said to be a cyclic vector.}, for
$\pi_\omega(A)\Omega_{\omega} = \psi_A$, but also that
$$
 \<\Omega_{\omega},\pi_\omega(A)\Omega_{\omega}\>_\omega
   = \<\psi_{\mathbbm{1}},\psi_{A}\>_\omega = \omega(A)
$$
for any $A \in \A$, \emph{i.e.}, within this particular representation
crafted for $\omega$, this state appears as a vector state
$\Omega_\omega$.

The triple $(\H_\omega,\pi_\omega,\Omega_\omega)$ is called the GNS
representation of the \Cstar-algebra $\A$ for the state $\omega$.

Let be $\A$ a \Cstar-algebra and $\omega$ a state on it. A triple
$(\H,\pi,\Omega)$ consisting of a Hilbert space $\H$, a representation
$\pi$ of $\A$ in $\L(\H)$ and a vector $\Omega \in \H$ such that the
set $\left\lbrace \pi(A)\Omega : A \in \A \right\rbrace$ is dense in
$\H$ and that $\<\Omega,\pi(A)\Omega\> = \omega(A)$ is called a {\em cyclic
representation} of $\A$ for the state $\omega$.

As we see, for any state $\omega$ over the \Cstar-algebra $\A$, there
exists at least one cyclic representation, the GNS one,
$\left(\H_{\omega},\pi_{\omega},\Omega_{\omega}\right)$.
Actually, it is not difficult to verify that the cyclic
representations are unique up to unitary equivalences. See, for
instance, Theorem 2.3.16 of \cite{bratteli} for details.

Yet, given that we ended up representing $\omega$ as a vector state
despite of it being mixed or not, we may ask: is there a way to
discover whether $\omega$, as a functional on $\A$, was a pure state or
a mixture by analysing its cyclic representations?

Now intervenes the notion of irreducibility.
A representation $\pi$ of an algebra $\A$ into a vector space $V$ is
said to be irreducible if there is no closed subspace $U \subset V$
left invariant by the action of $\pi$, apart from the trivial spaces
$U = \{0\}$ and $U = V$. Said otherwise, $\pi$ is irreducible if,
should $\pi(A)U \subset U$ for $U \subset V$ closed and every
$A \in \A$, then either $U = \{0\}$ or $U = V$. If $\pi$ is not irreducible,
it is said to be reducible.

The following theorem is of central importance for this discussion.
\begin{theorem}\label{th:reducibility}
  Let $\A$ be a \Cstar-algebra and $\omega$ a state on it. Then,
  $\omega$ is pure if, and only if, given a cyclic representation
  $(\H,\pi,\Omega)$ of $\A$ for $\omega$, $\pi$ is
  irreducible.
\end{theorem}

A proof can be found in references like \cite{murphy,bratteli}. For
the convenience of the reader we present it in Appendix
\ref{sec:appendic:ProofofTheoremth:reducibilit}.

As we see, purity manifests itself in irreducibility of the cyclic
representation of the algebra. On the other hand, the fact that the
cyclic representation of a mixed state is reducible (and, therefore,
can be further decomposed into irreducible ones, as discussed in
Section \ref{sec:Krein-MilmanandChoquettheorems}) means that such
states can be interpreted as being built by elementary subsystems.
This result is of major importance for understanding the distinction
between pure states and mixtures.

\section{Purification of States}
\label{sec:Purification-of-States}

Given a quantum state defined in a separable Hilbert space
$\H_{\mathrm{I}}$ by a density matrix $\rho$ it is possible to find
another (not uniquely defined) separable Hilbert space
$\H_{\mathrm{II}}$ such that the original state can be represented as
a normalised \underline{vector} state $\Psi_\rho$ in the tensor
product space $\H_{\mathrm{I}}\otimes \H_{\mathrm{II}}$. Therefore,
the vector state defined by $\Psi_{\rho}$ is a pure state in the
enlarged algebra $\L\big(\H_{\mathrm{I}}\otimes\H_{\mathrm{II}}\big)$
but, except when $\rho$ is a one dimensional projection, it is not a
pure state for the original algebra
$\L\big(\H_{\mathrm{I}}\big)\simeq \L\big(\H_{\mathrm{I}}\big)\otimes
\UM_{\H_{\mathrm{II}}}$. The vector
$\Psi_\rho\in\H_{\mathrm{I}}\otimes \H_{\mathrm{II}}$ is called a {\em
  purification} of the state defined by $\rho$.

These facts, which we are going to establish bellow, have an interesting
physical interpretation, since they say that a given mixed state of a
quantum system can always be thought of as a pure state of a larger
quantum system.

The vector $\Psi_\rho$ is not to be confused with the GNS vector
presented above, since the vector $\Psi_\rho$ is not cyclic and
separating for the original algebra
$\L\big(\H_{\mathrm{I}}\big)\simeq\L\big(\H_{\mathrm{I}}\big)\otimes\UM_{\mathrm{II}}$
and since the algebra
$\L\big(\H_{\mathrm{I}}\otimes\H_{\mathrm{II}}\big)$ is indeed
``larger'' than the ``original'' algebra
$\L\big(\H_{\mathrm{I}}\big)\otimes\UM_{\mathrm{II}}$.

Let $\H_{\mathrm{I}}$ be a separable Hilbert space and let $\rho$
be a density matrix acting on $\H_{\mathrm{I}}$, {\em i.e.},
a bounded, trace class, self-adjoint and positive operator acting on
$\H_{\mathrm{I}}$ with $\Tr_{\H_{\mathrm{I}}}(\rho)=1$.
The expression $\L(\H_{\mathrm{I}})\ni D\mapsto
\Tr_{\H_{\mathrm{I}}}(\rho D)\in\C$ represents a state on the \Cstar-algebra
$\L(\H_{\mathrm{I}})$.

Consider any other separable Hilbert space $\H_{\mathrm{II}}$ with
$\mathrm{dim}(\H_{\mathrm{II}})\geq \mathrm{dim}(\H_{\mathrm{I}})$. We
can find a normalised vector
$\Psi_\rho\in \H := \H_{\mathrm{I}}\otimes \H_{\mathrm{II}}$ such that
\begin{equation}
\Tr_{\H_{\mathrm{I}}}(\rho D)
\; = \;
\biglan \Psi_\rho, \; (D \otimes \UM_{\H_{\mathrm{II}}}) \Psi_\rho \bigran_{\H}
\; = \;
\Tr_{\H}\big( P_{\Psi_\rho} (D \otimes \UM_{\H_{\mathrm{II}}}) \big)
\;,
\label{eq:FormulaDePurificacao-caso-infinito}
\end{equation}
where $P_{\Psi_\rho}$ is the projector on the subspace generated by $\Psi_\rho$.

Let $\varrho_k, \; k\in\N$ be the eigenvalues of $\rho$ (including
multiplicity), with $\varrho_k\geq 0$ for all $k$, and let 
$\mathbf{v}_k, \; k\in\N$ be the corresponding normalised eigenvectors, building a
complete orthonormal basis in $\H_{\mathrm{I}}$:
$\biglan\mathbf{v}_i, \; \mathbf{v}_j
\bigran_{\H_{\mathrm{I}}}=\delta_{i,j}$.  The spectral decomposition
of $\rho$ is
$\rho = \sum_{i=1}^\infty \varrho_i P_{\mathbf{v}_i}$, where
$P_{\mathbf{v}_i}$ is the orthogonal projector on the one dimensional
subspace generated by $\mathbf{v}_i$.

Let $\{\mathbf{w}_l, \; l\in\N\}$ be an arbitrary orthonormal set of vectors
(not necessarily complete)
in $\H_{\mathrm{II}}$ ({\em i.e.}, with
$\biglan\mathbf{w}_i, \; \mathbf{w}_j
\bigran_{\H_{\mathrm{II}}}=\delta_{i,j}$). Define
$\Psi_\rho\in \H$ by
\begin{equation}
 \Psi_\rho \; := \; \sum_{i=1}^\infty \sqrt{\varrho_i}  \big( \mathbf{v}_i\otimes\mathbf{w}_i\big)
\; .
 \label{eq:FormuladoVetorPurificador-caso-infinito}
\end{equation}
The sequence in the r.h.s.\  converges in $\H$, because
the vectors $\mathbf{v}_i\otimes\mathbf{w}_i$,  $i\in\N$, are
orthonormal and because $\{\sqrt{\varrho_i}, \; i\in\N\}$ is a square
summable numerical sequence, since
$\sum_{i=1}^\infty\varrho_i=\Tr_{\H_{\mathrm{I}}}(\rho)=1$.
We have,
\begin{multline*}
  \biglan \Psi_\rho, \; (D \otimes \UM_{\H_{\mathrm{II}}}) \Psi_\rho \bigran_{\H}
  \; = \;
  \sum_{i=1}^\infty\sum_{j=1}^\infty\sqrt{\varrho_i} \sqrt{\varrho_j} 
  \Biglan \mathbf{v}_i\otimes\mathbf{w}_i , \;
  \big(D\mathbf{v}_j\big)\otimes\mathbf{w}_j\Bigran_{\H}
  \\
 =\; \sum_{i=1}^\infty\sum_{j=1}^\infty\sqrt{\varrho_i} \sqrt{\varrho_j}  
  \biglan \mathbf{v}_i , \;  D\mathbf{v}_j\bigran_{\H_{\mathrm{I}}}
  \underbrace{\biglan \mathbf{w}_i , \;
         \mathbf{w}_j\bigran_{\H_{\mathrm{II}}}}_{\delta_{i,j}}
  \; = \;
  \sum_{i=1}^\infty \varrho_i \biglan \mathbf{v}_i , \;
  D\mathbf{v}_i \bigran_{\H_{\mathrm{I}}}
  \;  = \; \Tr_{\H_{\mathrm{I}}} (\rho D) \;,
\end{multline*} 
establishing \eqref{eq:FormulaDePurificacao-caso-infinito}. For
$D=\UM_{\H_{\mathrm{I}}}$, in particular, this relation shows that
$\big\|\psi_\rho\big\|_{\H}=1$.

The vector $\Psi_\rho$ defined in
\eqref{eq:FormuladoVetorPurificador-caso-infinito} depends on 
$\rho$ (through the eigenvalues $\varrho_i$ and the eigenvectors
$\mathbf{v}_i$) and on the arbitrary choice of the orthonormal vectors
$\mathbf{w}_j$ of $\H_{\mathrm{II}}$. The l.h.s.\ of
\eqref{eq:FormulaDePurificacao-caso-infinito}, however, does not depend
on the choice of the $\mathbf{w}_j$'s.

It is relevant to notice that, except when $\rho$ is a one
dimensional projection, the vector state defined by $\Psi_\rho$
is not a pure state for the algebra
$\L\big(\H_{\mathrm{I}}\big)\otimes\UM_{\mathrm{II}}\simeq\L\big(\H_{\mathrm{I}}\big)$.
It is clear from 
\eqref{eq:FormulaDePurificacao-caso-infinito} that, for two
density matrices $\rho$ e $\rho'$ and for $\lambda\in[0, \; 1]$, one has
\begin{multline*}
\Biglan \Psi_{\lambda\rho+(1-\lambda)\rho'}, \;
  (D \otimes \UM_{\H_{\mathrm{II}}})
   \Psi_{\lambda\rho+(1-\lambda)\rho'} \Bigran_{\H}
\\ = \; 
\lambda
\biglan \Psi_{\rho}, \; (D \otimes \UM_{\H_{\mathrm{II}}}) \Psi_{\rho} \bigran_{\H}
+ (1-\lambda)
\biglan \Psi_{\rho'}, \; (D \otimes \UM_{\H_{\mathrm{II}}}) \Psi_{\rho'} \bigran_{\H}
\;.
\end{multline*}
From this, it also follows that
$$
\Biglan \Psi_\rho, \; (D \otimes \UM_{\H_{\mathrm{II}}}) \Psi_\rho \Bigran_{\H}
 \; = \; 
\sum_{i=1}^\infty
\varrho_i
\biglan \Psi_{P_{\mathbf{v}_i}}, \; (D \otimes \UM_{\H_{\mathrm{II}}}) \Psi_{P_{\mathbf{v}_i}} \bigran_{\H}
\;,
$$
since $\Psi_{P_{\mathbf{v}_i}}=\mathbf{v}_i\otimes\mathbf{w}_i$.
Hence, except when $\rho$ is a one dimensional projection, the vector
state defined by $\Psi_{\rho}$ is not a pure state for the algebra
$\L\big(\H_{\mathrm{I}}\big)\simeq \L\big(\H_{\mathrm{I}}\big)\otimes
\UM_{\H_{\mathrm{II}}}$.

One can easily show that
$ \rho = \PTr_{\H_\mathrm{II}}(P_{\Psi_\rho})$, where
$\PTr_{\H_\mathrm{II}}$ denotes the partial trace with respect to the
Hilbert space $\H_\mathrm{II}$.

A relevant question is whether a physical process (through a
completely positive map) can be identified leading to one of the
purifications associated to a given state. Such a process is known as
{\em physical purification} and we refer the reader to
\cite{KleinmannKampermannMeyerBruss} for further discussions on this
issue.


\section{Normal States and Density Matrices}\label{sec:normal_states}

We now come to the important relation between purity and normality of
states. In spite of being mathematically a more  technical discussion,
it is of central importance to physics due to its relation to the
notion of density matrices and to other issues.

The underlying question relevant to physics is: under which
circumstances can a state $\omega$ be defined by a density matrix,
{\em i.e.}, can be written in the form
$ \omega(A) = {\Tr}\,(\rho A) $?

Definition \ref{def:state} gives rise to far more exotic states than
those we are used to, represented by density operators when
$\A \subset \L(\H)$ for some separable Hilbert space $\H$. In this
case, these density operators amount to an important part of the set
of all states: as a corollary of Theorem \ref{th:normal_is_dense}, the
density operators are dense in the whole set of states with respect to
the weak operator topology, also known as physical topology (details
ahead). As a result, we are led to looking for further
characterisations regarding states in order to better understand them.

In particular, we will see that density operators are the realisation
of \emph{normal states} in the case where $\A \subset \L(\H)$.

Recall that an increasing net of operators is a net
$\left(A_\nu\right)_{\nu \in I} \subset \A$\footnote{$I$ is the
  \emph{index set}, the set where the indices take their values from.}
such that $A_\nu \geqslant A_\mu$ whenever $\nu \geqslant \mu$; the
operators' order relation is defined in the following way:
$A_\nu \geqslant A_\mu$ if $A_\nu - A_\mu$ is a positive operator,
\emph{i.e.}, if its spectrum is included in $\R^+$.

\begin{definition}\label{def:normal_state}
  A state (or any positive functional) $\omega$ on a \Cstar-algebra
  $\A$ is said to be normal if
$$\omega\left(\sup_{\nu \in I} A_\nu\right)= \sup_{\nu \in I} \omega(A_\nu)$$
for any bounded increasing net of positive operators $\left(A_\nu\right)_{\nu \in I} \subset \A$.
\end{definition}

As one can see, the definition of normal states heavily depends on the
notion of operator order and seems to express some kind of ``order
continuity'' property. Indeed, one may wonder whether there exists
some topology in $\A$ for which a state being normal would merely mean
that it is continuous. A topology as such does exist, actually more
than one. Let us define them.

\begin{definition}\label{def:topologies}
  Let $\H$ be a Hilbert space. We define the following families of
  seminorms on $\L(\H)$:
\begin{itemize}
\item \textbf{Ultra-strong seminorms.} Given $(\psi_n)_{n \in \N} \subset \ell_2(\H)$:
$$\|A\|_{\psi} = \left( \sum_{n=1}^\infty \|A\psi_n\|^2 \right)^\frac{1}{2}.$$
\item \textbf{Ultra-weak seminorms.} Given $(\psi_n)_{n \in \N}, (\phi_n)_\N \subset \ell_2(\H)$:
$$\| A \|_{\psi,\phi} = \left( \sum_{n=1}^\infty \left|\ip{\psi_n}{A\phi_n}\right|^2\right)^\frac{1}{2}.$$
\end{itemize}
The topologies induced on $\L(\H)$ by these families are known,
respectively, as the ultra-strong operator topology and the 
ultra-weak operator topology.
\end{definition}

Although these two topologies are distinct, being ultra-weakly or
ultra-strongly continuous is equivalent for linear
functionals.

\begin{theorem}\label{th:characterization_normal}	
  Let $\A \subset \L(\H)$ be a \Cstar-algebra, and $\omega$ a positive functional on it. The following conditions are equivalent:
\begin{enumerate}[(i)]
\item $\omega$ is normal;
\item $\omega$ is ultra-strongly continuous;
\item $\omega$ is ultra-weakly continuous;
\item there exists $(\phi_n)_{n\in\mathbb{N}}\subset \H$ with
  $\sum_{n=1}^{\infty}\|\phi_n\|^2 <\infty$ such that
  \begin{equation}
  \omega = \sum_{n=1}^{\infty}\omega_{\phi_n}, \quad
  \text{\footnotesize (convergence in norm),}
  \label{eq:normalidade}
\end{equation}
where $\omega_{\phi_n}$ is given, for any $A \in \A$, by $\omega_{\phi_n}(A) = \<\phi_n,\, A\phi_n\>$.
\end{enumerate}
\end{theorem}

A proof of this theorem can be found in \cite{bratteli,dixmier}.


If $\H$ is a separable Hilbert space, it is now easy to conclude from
(\ref{eq:normalidade}) that a normal state on $\A \subset \L(\H)$ can
be represented by
\begin{equation}\label{eq:state_density_operator}
\omega(A) = {\Tr}\,(\rho A), \qquad \forall A \in \A,
\end{equation}
where
$$\rho = \sum_{n = 1}^\infty \|\phi_n\|^2 \left| \Phi_n \>\< \Phi_n \right|,$$
$(\phi_n)_{n\in\mathbb{N}}\subset\H$ being the sequence obtained in
the above theorem, and $\Phi_n = \frac{\phi_n}{\|\phi_n\|}$.

Besides, since a state is by hypothesis normalised, it follows that we
will also have
${\Tr}\, (\rho) = \sum_{n=1}^\infty \|\phi_n\|^2 = 1$, so $\rho$ is
a genuine density operator.

\begin{theorem}\label{th:normal_density_operator}
  Let $\H$ be separable. Normal states are realised by density
  operators when $\A$ is mapped into $\L(\H)$. Conversely, any states
  realised by density operator is a normal state.
\end{theorem}

\begin{proof}
  From Theorem \ref{th:algebra_isomorphism} it is known that $\A$ may
  be isomorphically mapped onto a closed subalgebra
  $\tilde{\A} \subset \L(\H)$, so the states on $\A$ may be mapped
  onto the states on $\tilde{\A}$. By Theorem
  \ref{th:characterization_normal} and the subsequent discussion, the
  proof is complete.
\end{proof}

Now, this section's central result:

\begin{theorem}\label{th:normal_is_dense}
  Normal states are dense in the set of all states according to the the weak (physical)
  topology, \emph{i.e.}, given any state $\omega$ on a \Cstar-algebra
  $\A$, there exists a sequence of normal states
  $(\omega_n)_{n \in \N} \subset \A$ such that, for any $A \in \A$,
  $\omega(A) = \lim_n \omega_n (A)$.
\end{theorem}

\begin{proof}
Theorem 1.1 in \cite{fell}.
\end{proof}

\begin{remark}
  As we saw in Section \ref{sec:physical_topology}, an actual state
  $\omega$ of a physical system may be determined by a process of
  taking limits within the physical topology. In view of the previous
  theorem, we may only consider normal states to partially describe
  $\omega$ within weak neighbourhood
  $\Omega_\eps(A_1,\, \ldots ,\, A_k;\omega)$; however, it is important to bear
  in mind that, even in this case, the limit, \emph{i.e.}, $\omega$,
  may end up being non-normal.
\end{remark}

If a state is normal, it is easier to identify it as pure or mixed:

\begin{theorem}\label{th:characterization_normal_pure_mixture}
  Let $\A$ be a \Cstar-algebra isomorphic to
  $\tilde{\A} \subset \L(\H)$ for some separable Hilbert space
  $\H$. Let $\omega$ be a normal state on $\A$ and $\rho \in \tilde{\A}$
  the density operator associated to $\omega$ (see Theorem
  \ref{th:normal_density_operator} above). Then $\rho$ has
  Hilbert-Schmidt norm
  $\|\rho\|_{HS} = \sqrt{{\Tr}(\rho^\ast\rho)} = 1$ if and only if
  $\omega$ is pure. If $\omega$ is a mixed state, $\|\rho\|_{HS} < 1$.
\end{theorem}

\begin{proof}
  Since $\rho$ is a density operator, it is positive and compact, so
  it is possible to choose a Hilbertian basis
  $(e_n)_{n \in \N} \subset \H$ and a sequence of scalars
  $\varrho_n \geqslant 0$ such that
  $\rho = \sum_{n = 1}^\infty \varrho_n \left| e_n \>\<e_n\right|$. Given
  that ${\Tr}\, (\rho) = \sum_{n = 1}^\infty \varrho_n = 1$, we must have
  $\varrho_n \leqslant 1$ for every $n \in \N$, which implies
  $\varrho_n^2 \leqslant \varrho_n$, the equality holding only in the case 
  $\varrho_n = 0$ or $\varrho_n = 1$; as a consequence, calculating its
  Hilbert-Schmidt norm:
$$\|\rho\|_{HS}^2 = {\Tr}\, (\rho^* \rho) = \sum_{n = 1}^\infty \varrho_n^2,$$
it becomes clear that $\|\rho\|_{HS}$ equals $1$ if and only if
there is exactly one $\varrho_n$ non-null, say for $n = N$. Thus, we have
$\rho = \left| e_N \> \< e_N \right|\in \A$, which means that $\omega$
is mapped to a state that acts on $A \in \tilde{\A}$ as
${\Tr}\,(\rho A) = \<e_N,\, Ae_N\>$, a vector state. By Theorem
\ref{th:vector_is_pure-generalisation}, states of this form are always
pure, implying the same for $\omega$ itself. Conversely, if
$\|\rho\|_{HS} < 1$ strictly, then we have at least two non-zero
$\varrho_n$'s; calling one of them $\varrho_N$, one has:
$$\rho = \lambda \left| e_N \> \< e_N \right| + (1 - \lambda) \rho^\prime,$$
with $0 < \lambda = \varrho_N < 1$ and
$\rho^\prime = \frac{1}{1-\varrho_N} \sum_{\underset{n \neq N}{n = 1}}^\infty \varrho_n
\left| e_n \> \< e_n \right|$. Since $\rho^\prime$ also represents a state,
for it is a density operator, we have that $\rho$ is a non-trivial
mixture, so it is not pure.

The converse affirmations, \emph{i.e.}, that being pure implies
$\|\rho\|_{HS} = 1$, and being a mixture $\|\rho\|_{HS} < 1$, are also
true, as any density operator has Hilbert-Schmidt norm within $(0,1]$.
\end{proof}

As a side remark, we would like to highlight that it is not trivial to
see that $\left| e_N \> \< e_N \right|$ and $\rho^\prime$ define
states on $\tilde{\A}$, since the spectral projections of $\rho$ have
no reason to belong to $\tilde{\A}$. Fortunately, the spectral
projections are elements of the closure of $\tilde{\A}\subset\L(\H)$
in the weak operator topology (wot), hence
$\left| e_N \> \< e_N \right|$ and $\rho^\prime$ define normal states
in the von Neumann algebra $\overline{\tilde{\A}}^{wot}$, the closure
of $\tilde{\A}$ in the weak operator topology. Finally, as seen in
Theorem \ref{th:characterization_normal}, normal operators are
(ultra-)weakly continuous and the conclusion holds.

We should emphasise that normal states may be pure or mixed. From
their characterisation given in Theorem
\ref{th:characterization_normal_pure_mixture}, normal pure states
acting on $\nalgebra \subset \L(\H)$ are density operators on $\H$ with
unitary Hilbert-Schmidt norm, and from the very proof of this result
one sees that:

\begin{theorem}\label{th:normal_and_pure_vector}
Any normal pure state on the \Cstar-algebra $\L(\H)$ is a vector state.
\end{theorem}

This last result justifies a very common statement, found in many
Quantum Mechanics textbooks, that pure states are those whose Hilbert-Schmidt
norm equals $1$, and mixtures those with $\|\rho\|_{HS} < 1$. This is
only correct for normal states, which tells us not the whole picture,
as we will just see in next section.






\section{More Issues About Purity}\label{sec:moreissuesaboutpurity}

In this section we discuss some important issues and examples
concerning the relation between pure and vector states.

\subsection{A pure state that is not a vector state}\label{sec:example_not_vector}
Here we will extend the study in Section \ref{sec:pure_and_vector} by
exhibiting an example of a pure state that is not a vector state. For
simplicity, let us suppose that $\H$ is a separable
infinite-dimensional Hilbert space with orthonormal basis
$\{e_n\}_{n\in\N }$. Let $(a_n)_{n\in\N }\subset [0,1)$ be a sequence
such that $a_n \underset{n \rightarrow \infty}{\longrightarrow} 1$, for
instance: $a_n=2^{-\frac{1}{n}}$.

Now, define an operator $A \in \L(\H)$ acting on a vector $\Psi = \sum_{n = 1}^\infty \psi_n e_n \in \H$ as:
$$A\Psi = \sum_{n\in \N} a_n \, \psi_n e_n.$$
Clearly, $\|A\Psi\| < \| \Psi \|$ for any $\Psi \in \H$, and together
with $\|Ae_n\| \longrightarrow 1$, we obtain $\|A\| = 1$. Finally,
Theorem 5.1.11 in \cite{murphy} states that there exists a pure state
$\omega$ on $\L(\H)$ such that $\omega(A) = \|A\| = 1$; we claim that
this pure state cannot be vector. In fact, if it were the case, we
would have, for some $\phi \in \H$ with $\|\phi\| = 1$:
$$1 = \omega(A) = \<\phi,\, A\phi\> \leqslant \|A\phi\| < 1,$$
which is an absurd.


\subsection{All pure states on $\K(\H)$ are vector states}
\label{sec:example_pure_is_vector}
The present example (inspired in Section 5.1.1 of \cite{murphy})
contrasts with the previous one. Above, we have shown that there may
be in general pure states which are not vector; here we will see the
opposite, \emph{i.e.}, a special case where all pure states are also vector,
stressing the importance of the particular algebra that we take for
observables.

To begin with, noting by $\K(\H)$ the algebra of compact operators on
a separable Hilbert space $\H$, it is known that its dual is composed
by the set of trace class operators acting on $\H$, in symbols:
$\K(\H) = \L_1(\H)$. For the reader's convenience, let us quickly
proof this fact by remarking that the linear function
$\L_1(\H) \ni A \longmapsto {\Tr}_A \in \K(\H)^*$, where
${\Tr}_A (K) = {\Tr}\ (AK)$ for any $K \in \K(\H)$, is an
isometric isomorphism.

Indeed, taking an element $\omega \in \K(\H)^*$, Riesz's
representation theorem implies that there is a bounded operator
$A_\omega \in \L(\H)$ such that the sesquilinear form
$$\H \times \H \ni (x,y) \longmapsto \omega\left( \left| x \>\< y \right| \right) \in \C$$
can be written as
$\omega\left( \left| x \>\< y \right| \right) = \<x,\, A_\omega
y\>$. $A_\omega$ is trace-class, for picking up a Hilbertian basis
$\{e_n\}_{n \in \N}$ of $\H$:
$$
 {\Tr}\ (A_\omega) = \sum_{n = 1}^\infty \<e_n,\, A_\omega e_n\> 
 = \sum_{n = 1}^\infty \omega \left( \left| e_n \>\< e_n \right| \right) 
 = \sum_{n = 1}^\infty \omega \left( \mathbbm{1} \right) \leqslant \|\omega\|.
$$
Using continuity and linearity of $\omega$, denseness of the
finite-rank operators in $K(\H)$ and further remarks about the
injectivity $A \mapsto {\Tr}_A$, we obtain the desired duality.

Now, concerning a state $\omega$ on $\K(\H)$, it is easy to see that
the corresponding $A_\omega$ will be positive; since it is also
compact (as trace-class implies compact), there exists a Hilbertian
basis $\{e_n\}_{n \in \N}$ of $\H$ for which $A_\omega$ is diagonal,
\emph{i.e.} $A_\omega e_n = \lambda_n e_n$, with
$\lambda_n \geqslant 0$, and $\sum_{n = 1}^\infty \lambda_n = 1$ (this
sum comes from the normalisation of $\omega$). For $K \in \K(\H)$, we
have:
\begin{equation}\label{eq:x2}
\omega(K) = {\Tr}\ (A_\omega K) = \sum_{n = 1}^\infty \lambda_n \ip{e_n}{K e_n}.
\end{equation}
Hence, any state $\omega$ on $\K(\H)$ is a convex combination of
states like $\<e_n,Ke_n\>$; if $\omega$ is pure, then $\lambda_N=1$
for some $N \in \N$ and $\lambda_n=0$ for $n\neq N$, so we conclude
that it is also a vector state.



\section{Krein-Milman's and Choquet's theorems}
\label{sec:Krein-MilmanandChoquettheorems}

From the very beginning, we have been talking about pure states, but
until now we have not answered a crucial question: do they exist?

Notice that pure states constitute some kind of ``fundamental brick''
in the construction of states, that is, if we have a state $\omega$ we
can wonder if it is a mixed state. Then, if it is a mixed state, we
have $\omega=\lambda_1 \omega_1+\lambda_2 \omega_2$ for two distinct
states $\omega_1$ and $\omega_2$ and $\lambda_1,\lambda_2\in (0,1)$,
and we can wonder now if $\omega_1$ and $\omega_2$ are themselves
mixed states. Proceeding this way, after some steps, we write the
original state as a convex combination
$\omega=\sum_{i=1}^n{\lambda_i \omega_i}$, with $\lambda_i\in(0,1)$
and $\sum_{i=1}^n \lambda_i =1$. This procedure is very similar to the
one used to prove that a positive integer number has a prime
decomposition, but there is a very important difference because you
cannot divide positive integers forever by divisors bigger than
one. This difference creates the possibility that the process we
suggested for decomposing states never stops, leaving unanswered the
question about the very existence of pure states (apart from some
concrete examples, like the cases where $\A = \mathcal{L}(\H)$,
where it is known that vector states are pure).

Is there a way to circumvent the problem of our infinite process
appealing for topology, that means, can we assure that at least the
sequence obtained by our steps is convergent, in which case we could
write $\omega=\sum_{i=1}^\infty {\lambda_i \omega_i}$, with
$\lambda_i\in(0,1)$ and $\sum_{i=1}^\infty \lambda_i =1$?

Fortunately, it is possible to prove that there exist pure states
and, as our previous discussion suggests, that they exist in such a
number that all states can be written as limits (in a suitable
topology) of convex combination of them. This is Krein-Milman's
theorem.

In order to state Krein-Milman's theorem, we need to define what is a
face and what is an extremal point.

\begin{definition}
  Let $V$ be a topological vector space and $C\subset V$ be a
  non-empty convex subset. A non-empty closed and convex subset
  $F\subset C$ is said to be a face (or extremal) set of $C$ if, given
  $x,y \in C$ and $\lambda \in (0,1)$, the fact that
  $\lambda x +(1-\lambda)y \in F$ imply $x,y \in F$.
	
\end{definition}

Notice that a face is a set such that, if it contains any internal
point of a line segment of $C$, then it contains the whole segment. A
good intuition on this definition comes from polyhedra, which are in
fact, the origin of the name ``face''. If we think of a cube, its
squared faces are, indeed, six faces in the sense above.

Notice now that a face is again a non-empty convex set, hence we can
ask about the faces of a face. It is not difficult to notice that a
face of a face is also a face of the original set (see Lemma 2.10.5 of
\cite{megginson}). Back to our example, we can ask about the faces of
the six squares and, it is easy to verify, they are the squares' edges
(and the cube's vertices), and the edges' faces are the ending points
of the edges, \emph{i.e.}, faces without subfaces that are unitary
sets containing each of the cube's vertices. Let us give a special
name to these points.

\begin{definition}
  Let $V$ be a Hausdorff topological vector space and let $C\subset V$
  be a non-empty convex set. An extremal point of $C$ is a element
  $x\in C$ such that $\{x\}$ is a face of $C$.
	We denote ${\rm Ext} (C) = \{x \in C  \ : \ x \textrm{ is an extremal point of C}\}$.
\end{definition}

Of course the cube is a very simple instance of the general question,
but this example gives us a general idea on what is going on: the
extremal points of the cube, namely, its vertices, can be used to
obtain any other of the cube's points by taking convex combinations:
first obtaining the edges, after the faces and finally the interior of
the cube. That is, the cube is the smallest convex set containing its
vertices. We call this smallest convex set the convex hull, that is,
the convex hull of a set $A$ is the intersection of all convex subsets
of the vector space containing $A$. The convex hull of a set $A$ is
denoted by $\co{A}$. An analogous definition can be done by taking the
closed convex subsets, that is the closed convex hull, which is
denoted by $\cco{A}$.  \linebreak
\begin{minipage}{0.75 \textwidth}

\hspace{12pt}

Another interesting fact is that, even in finite dimension,
${\rm Ext}(K)$ is not necessarily closed. Consider for example the
$y$-displaced double cone in $\mathbb{R}^3$, {\small
$$C=\cco{\{(x,y,0) \in \mathbb{R}^3 \ : \ x^2+(y-1)^2=1\}\cup \{(0,0,1),(0,0,-1)\}}.$$}
Notice that $(0,0,0)$ cannot be an extremal point of $C$, since \linebreak$(0,0,0)=\frac{1}{2}(0,0,1)+\frac{1}{2}(0,0,-1)$. In fact,
{\small
$${\rm Ext}(C)=\{(x,y,0) \in \mathbb{R}^3 \ : \ x^2+(y-1)^2=1, \  x\neq0 \}\cup \{(0,0,1),(0,0,-1)\},$$}
which is not closed.
\end{minipage}
\begin{minipage}{0.25 \textwidth}
	\includegraphics[width=1\textwidth]{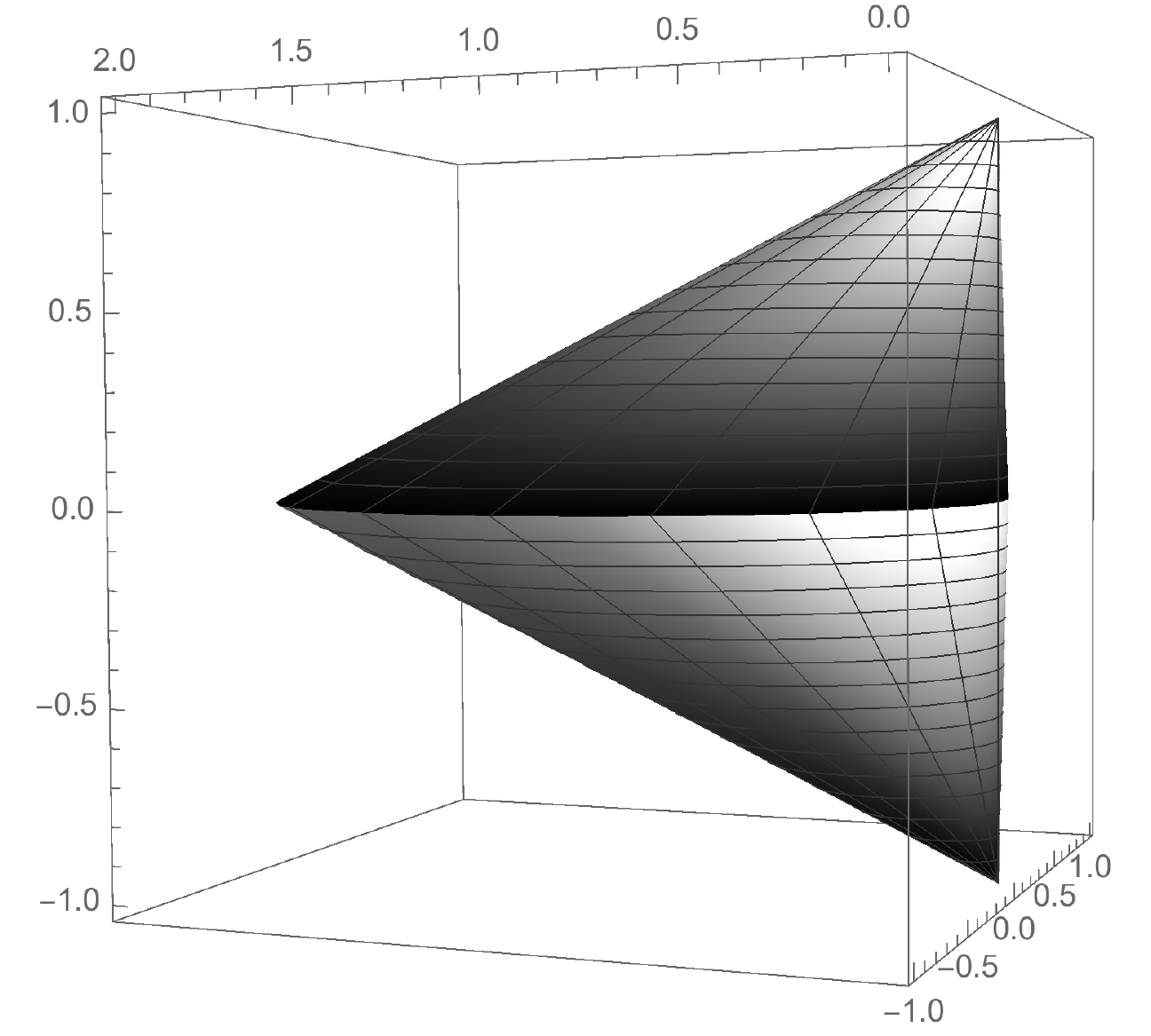} 
	\centering\footnotesize {The double cone.}
\end{minipage}
\vspace{1pt}

Finally, we are fit for stating the general result.

\begin{theorem}[Krein-Milman]
	\label{TKM}
	Let $V$ be a Hausdorff locally convex topological vector space
        and let $K\subset V$ be compact and convex. Then
        $K=\cco{{\rm Ext}(K)}$.
\end{theorem}

It is important to reinforce that compactness plays a central role in
the proof of Krein-Milman's theorem. In fact, we need compactness to
ensure that this process of taking faces of faces does end and does
not end up in an empty set, which is similar to the issue described in
this section's second paragraph.

Now that we have Krein-Milman's theorem in hand, we can clarify our
previous claim on the existence of pure states. First of all, the
closed unit ball of any infinity dimensional Banach space is not
compact in the norm topology, but the closed unit ball in the dual of
a normed vector space is compact in the weak-$\ast$ topology, thanks
to the Banach-Alaoglu's theorem (see Theorem 2.6.18 of
\cite{megginson}). The conclusion is that the states of any
\Cstar-algebra $\A\subset \mathcal{L}(\H)$ constitute a weak-$\ast$
closed subset $\mathcal{S}$ of the weak-$\ast$ compact set of all
bounded and normalised linear functionals on $\A$. Hence,
$\mathcal{S}$ is weak-$\ast$ compact itself and, by Krein-Milman's
theorem, $\mathcal{S}=\cco{{\rm Ext}\left(\mathcal{S}\right)}$. Since
the extremal points of $\mathcal{S}$ are the states that do not lay in
the interior of any segment line of $\mathcal{S}$, what is a way to
say that they are not convex combinations of other states,
${\rm Ext}\left (\mathcal{S}\right)$ is the set of all pure states.

There are several interesting consequences of Krein-Milman's theorem,
among them, it can be used to prove that given $A\in\A$, there is a
pure state $\omega_A$ such that $\omega_A(A)=\|A\|$. This result is
used in the construction presented in Section
\ref{sec:example_not_vector}.

Let us return to Krein-Milman's theorem. Suppose $V$ is a Hausdorff
locally convex topological vector space and $K\subset V$ is compact
and convex. Notice now that, for every $x \in \co{{\rm Ext}(K)}$,
there exists $n\in \mathbb{N}$, $\{\lambda_i\}_{i=1}^n\subset [0, \; 1]$, with $\sum_{i=1}^n\lambda_i=1$, and
$\{x_i\}_{i=1}^n\subset {\rm Ext}(K)$ such that $x=\sum_{i=1}^n \lambda_i x_i$. Then, for
every continuous linear functional $f$ on $V$,
$f(x)=f(\sum_{i=1}^n \lambda_i x_i)=\sum_{i=1}^n \lambda_i
f(x_i)$. When ${\rm Ext}(K)$ is finite, we could also write
$f(x)=\sum_{e\in {\rm Ext}(K)} \lambda_e f(e)$, for every
$f\in V^\ast$, with $\lambda_e\in[0, \; 1]$. Notice that the different $\lambda_e$'s work as weights
in the sum, with $\sum_{e\in {\rm Ext}(K)}\lambda_e=1$. We can define
the probability measures (Dirac measures) $\delta_e:K \to \mathbb{R}$
by
$$
  \delta_e(X)=\begin{cases} 1, & \textrm{ if } e\in X, \\ 0, & \textrm{ if } e\notin X,\end{cases}
$$
for each $e\in {\rm Ext}(K)$ and
$\mu=\sum_{e\in {\rm Ex}(K)} \lambda_e \delta_e$. It is quite easy to
check that $f(x)=\int_K f d\mu$.

Notice that the measure $\mu$ is a regular Borel measure satisfying
$\mu(K)=1$, in other words, $\mu$ is a probability measure. Further,
$\mu\big(K\setminus {\rm Ext}(K)\big)=0$. A measure satisfying this is said to
be supported in ${\rm Ext}(K)$. In addition, we have that
$f(x)=\int_K f d\mu$ for all continuous linear functional on $V$. A
measure satisfying this property is said to represent $x\in K$

This is not an isolated case, and its generalisation is given by: 
\begin{theorem}[Choquet]
	\label{choquet}
        Let $V$ be a Hausdorff locally convex space and $K\in V$ a
        metrizable convex compact set. Then, for every $x\in K$, there
        exists a probability measure $\mu$ supported in ${\rm Ext}(K)$
        representing $x$.
\end{theorem}

Krein-Milman's and Choquet's theorems are equivalent when ${\rm Ext}(K)$ is closed.


Choquet's theorem has a very interesting consequence in von Neumann
algebras. Let $\M\in \mathcal{L}(H)$ be a von Neumann algebra. Since a
von Neumann algebras is a dual space of some Banach space $\M_\ast$,
its closed unit ball is weak-$\ast$ compact. In addition, if $\H$ is
separable, the closed unit ball $B_1 \subset \mathcal{L}(\H)$ with the
weak-$\ast$ topology is metrizable. Hence, for $K=B_1$, we are in the
conditions of Choquet's theorem, so for any $A\in B_1$ there exists a
probability measure $\mu_A$ supported in ${\rm Ext}(B_1)$ such that,
for all continuous linear functionals on $\mathcal{L}(\H)$,
$f(A)=\int_{B_1} f d\mu_A$.

\begin{appendix}
\section{Proof of Theorem \ref{th:reducibility}}
\label{sec:appendic:ProofofTheoremth:reducibilit}

Here we follow \cite{murphy} closely.  First, take a non-null positive
linear functional $\omega'$ such that, $\forall A \in \A$,
$\omega'(A^*A) \leqslant \omega(A^*A)$. By the Cauchy-Schwartz inequality
\eqref{eq:cauchy_schwartz},
$$|\omega'(A^*B)|^2 \leqslant \omega'(A^*A)\omega'(B^*B) \leqslant \omega(A^*A)\omega(B^*B) = \|\pi(A)\Omega\|^2\|\pi(B)\Omega\|^2,$$
which implies, by Riesz representation theorem (see {\em e.g.}, \cite{reed_simon}),
the existence of a positive operator $T \in \L(\H)$ such that
$\<\pi(A)\Omega,T\pi(B)\Omega\> = \omega'(A^*B)$ (remark that the
sesquilinear form $(\pi(A)\Omega,\pi(B)\Omega) \longmapsto \omega'(A^*B)$
may be extended to the whole $\H$, and thus the domain of $T$,
only because $\Omega$ is cyclic). Taking $A,B,C \in \A$ arbitrary:
$$
\<\pi(A)\Omega,T\pi(B)\pi(C)\Omega\> = \omega'(A^*BC)
  = \omega'\big((B^*A)^*C\big) = \biglan \pi(A)\Omega,\pi(B)T\pi(C)\Omega\bigran,
$$
which implies that $[T,\, \pi(B)] = 0$ for any $B$.

As known in representation theory (as a consequence of Schur's Lemma,
see \emph{e.g.}  \cite{representation}), if $\pi$ is an irreducible
representation, any self-adjoint operator commuting with it must be a
multiple of the identity and vice-versa. It happens that, if
$T = \lambda \mathbbm{1}$ for some $\lambda \in \C$, then
$\omega' = \lambda \omega$. Let us show an implication of this fact, that
$\omega$ is a mixture if and only if $\pi$ is reducible, which is
enough for the theorem's statement.

Indeed, if $\omega$ is a mixture, one may find $\omega'$ such that
$\omega'(A^*A) \leqslant \omega(A^*A)$ which is not a multiple of
$\omega$: in this case, there is a scalar $\sigma \in (0,1)$ and
states $\omega_1$ and $\omega_2$ (none of them multiples of $\omega$) such
that $\omega = \sigma \omega_1 + (1-\sigma) \omega_2$, so just take
$\omega' = \sigma \omega_1$. The corresponding $T$ will not be a multiple
of the identity, hence the reducibility of $\pi$.

Conversely, supposing that $\pi$ is reducible, one may find a
self-adjoint $S \in \A$ commuting with every $\pi(B)$ and not being a
multiple of $\mathbbm{1}$. As a consequence, any non-trivial spectral
projector $P$ of $S$ will also commute with $\pi$, not be a multiple
of the identity, and further satisfy $0 < P < \mathbbm{1}$. Define the
functional $\omega'(A) = \<\Omega,\, P\pi(A)\Omega\>$ and remark that it is
positive, not a multiple of $\omega$, and that
$$
\omega(A^*A)-\omega'(A^*A) = \biglan \pi(A)\Omega,\, (\mathbbm{1}-P)\pi(A)\Omega\bigran \geqslant 0.
$$
This implies that $\omega$ is a non-trivial mixture
$\omega = \lambda \omega_1 + (1-\lambda)\omega_2$, with
$\lambda = \|\omega'\| \in (0,1)$ and states
$\omega_1 = \frac{1}{\|\omega'\|} \omega'$ and
$\omega_2 = \frac{1}{1-\|\omega'\|}(\omega - \omega')$ (see Remark
\ref{rem:ppc} for a quick justification that $\|\omega_2\| = 1$).

\end{appendix}
  

\thanks{{\bf Acknowledgements.} We are indebted to K.-H.\ Neeb for
  pointing us some incorrections in a previous version of this manuscript.}

\end{document}